\documentclass[a4paper,reqno,10pt]{amsart}
\makeatletter
\newcommand*{\rom}[1]{\expandafter\@slowromancap\romannumeral #1@}
\makeatother
\usepackage{epsf,graphicx,epsfig}
\usepackage{amscd}
\usepackage{amsmath,latexsym,amssymb,amsthm}
\usepackage{fdsymbol}
\usepackage[nospace,noadjust]{cite}
\usepackage{textcomp}
\usepackage{setspace,cite}
\usepackage{mathrsfs,amsmath}
\usepackage{chemarrow}
\usepackage{lscape,fancyhdr,fancybox}
\usepackage{stmaryrd}

\usepackage[all,cmtip]{xy}
\usepackage{tikz-cd}
\usepackage{cancel}

\usepackage{graphicx}
\usepackage[usestackEOL]{stackengine}
\usetikzlibrary{shapes, arrows, decorations.markings}
\usepackage{graphicx}

\usepackage{color}
\usepackage{bbm, dsfont}
\usepackage{xcolor}
\usepackage{url}
\usepackage{enumerate}
\usepackage[mathscr]{euscript}

  \theoremstyle{plain}
\swapnumbers
    \newtheorem{thm}{Theorem}[section]
    \newtheorem{proposition}[thm]{Proposition}
     \newtheorem{theorem}[thm]{Theorem}
   \newtheorem{lemma}[thm]{Lemma}

    \newtheorem{subsec}[thm]{}
\theoremstyle{definition}
    \newtheorem{definition}[thm]{Definition}
        \newtheorem{remark}[thm]{Remark}
    \newtheorem{exam}[thm]{Example}

\theoremstyle{remark}

\newcommand{\Z}{\mathbb{Z}}

\newcommand{\C}{\mathbb{C}}
\newcommand{\GG}{\mathcal{G}}
\newcommand{\D}{\mathcal{D}}
\newcommand{\M}{\mathcal{M}}
\newcommand{\E}{\mathcal{E}}
\newcommand{\p}{\partial}
\newcommand{\la}{\lambda}

\title{}
\author{}
\date{}
\usepackage{amssymb}

\usepackage{hyperref}
\hypersetup{
colorlinks,
citecolor=blue,
filecolor=black,
linkcolor=blue,
urlcolor=black
}

\usepackage[left=25mm,right=25mm,top=25mm,bottom=25mm,paper=a4paper]{geometry}

\begin{document}



\title[Diassociative conformal algebras and their deformation cohomology]{Diassociative conformal algebras and their deformation cohomology}

\author{Anupam Sahoo}
\address{Department of Mathematics,
Indian Institute of Technology, Kharagpur 721302, West Bengal, India.}
\email{anupamsahoo23@gmail.com, anupamsahoo135@kgpian.iitkgp.ac.in}

\author{Apurba Das}
\address{Department of Mathematics,
Indian Institute of Technology, Kharagpur 721302, West Bengal, India.}
\email{apurbadas348@gmail.com, apurbadas348@maths.iitkgp.ac.in}

\begin{abstract}
In this paper, we introduce diassociative conformal algebras that unify both associative conformal algebras and diassociative algebras. We show that the category of diassociative conformal algebras is equivalent to the category of equivalence classes of formal distribution diassociative algebras. By employing its Maurer-Cartan characterization, we then define the cohomology of a diassociative conformal algebra and show that it carries a Gerstenhaber algebra structure. Our cohomology can be used to study deformations and extensions of the structure. Finally, we define averaging operators on associative conformal algebras and find their close relations with diassociative conformal algebras.
\end{abstract}

\maketitle

\begin{center}
    {\em To the memory of Dr. Anita Naolekar (Majumdar)}
\end{center}

\medskip

\medskip



\noindent

\thispagestyle{empty}

\tableofcontents


\section{Introduction}\label{sec1}
Conformal algebras, introduced by Kac, provide an algebraic framework for the operator product expansion of chiral fields in conformal field theory \cite{ bor, bpz, kac-1}. In this regard, the notion of Lie conformal algebras was introduced as a $\C[\p]$-module $\mathfrak{g}$ equipped with a {conformal sesquilinear} map (called the $\la$-bracket) $[\cdot_{\la}\cdot]: \mathfrak{g} \otimes \mathfrak{g} \to \mathfrak{g}[\la]$ satisfying the following version of the skew-symmetry and the conformal Jacobi identity:
\begin{align*}
[a_{\la}b] = -[b_{-\la -\p}a] ~~ \text{ and } ~~  [a_{\la}[b_{\mu}c]]=[[a_{\la}b]_{\la +\mu} c] +[b_{\mu}[a_{\la} c]], \text{ for } a,b,c \in \mathfrak{g}.
\end{align*}
Lie conformal algebras play an important role in studying infinite-dimensional Lie algebras satisfying the locality property \cite{andrea-kac, kac-2}. More precisely, the category of finite Lie conformal algebras is equivalent to the category of infinite-dimensional Lie algebras spanned by Fourier coefficients of a finite number of pairwise local fields (or formal distributions) that are closed under the operator product expansion.

The notion of associative conformal algebras naturally arises from the study of representations of Lie conformal algebras. Explicitly, an associative conformal algebra is given by a $\C[\p]$-module $\mathcal{A}$ equipped with a conformal sesquilinear map (called the $\la$-product) $\cdot_{\la} \cdot: \mathcal{A}\otimes \mathcal{A} \to \mathcal{A}[\la]$ satisfying the conformal associativity condition:
\begin{align*}
 a_{\la}(b_{\mu}c)= (a_{\la}b)_{\la + \mu } c, \text{ for }a, b, c \in \mathcal{A}.
\end{align*}
This structure induces a Lie conformal algebra structure on $\mathcal{A}$ via the  $\la$-bracket $[a_{\la}b] = a_{\la}b - b_{-\la - \p} a$. In addition to the conformal associativity, several other identities hold in any associative conformal algebra $(\mathcal{A}, \cdot_{\la} \cdot)$, for example, 
\begin{align}
    a_{\la}(b_{-\p -\mu}c) =~& (a_{\la}b)_{-\p -\mu}c,\label{equation a}\\
    a_{-\p -\la}(b_{\mu}c) =~& (a_{-\p -\mu}b)_{-\p +\mu -\la }c,\label{equation b}\\
    a_{-\p -\la}(b_{-\p -\mu}c) =~& (a_{-\p +\mu -\la}b)_{-\p -\mu }c \label{equation c},
\end{align} 
for all $a,b,c \in \mathcal{A}$ \cite{andrea-kac}. Structure theory, representations and cohomology of Lie conformal algebras and associative conformal algebras are explored in several articles \cite{andrea-kac, fattori-kac, hong-su,Kolesnikov, roitman-1}. In \cite{bakalov-kac-voronov}, the authors used the cohomology of conformal algebras to study their extensions and deformations. Viewing conformal algebras as a generalization of classical algebras, many studies on conformal algebras have been developed inspired by the classical cases. For example, the authors in \cite{Hou-Shen-Zhao} have recently found a Gerstenhaber algebra structure on the cohomology of an associative conformal algebra. That is, there is a degree $0$ graded commutative associative product and a degree $-1$ graded Lie bracket that are compatible in terms of a Leibniz rule. This result is a generalization from the classical context \cite{gers-ring,gers-voro}. See also \cite{bakalov-kac-voronov, bpz, andrea-kac,fattori-kac,kac-2,hong-su,hong-bai,Kolesnikov,koli-Sart,roitman-1,sole-kac} for some other results about Lie and associative conformal algebras. On the other hand, the notion of Leibniz conformal algebras was introduced in \cite{bakalov-kac-voronov} as a noncommutative generalization of Lie conformal algebras. In \cite{zhang}, the author has shown that the category of Leibniz conformal algebras is equivalent to the category of equivalence classes of formal distribution Leibniz algebras.

\medskip

 In the study of classical algebraic structures, the notion of diassociative algebras was introduced by Loday in his study of Leibniz algebras \cite{loday1}. A diassociative algebra is a triple $(D,\dashv,~ \! \vdash)$ consisting of a vector space $D$ equipped with two bilinear operations $\dashv,~ \! \vdash: D \times D \to D $ such that for all $a, b, c \in D$, the following identities hold:
\begin{align}
a \dashv (b \dashv c) =~& (a \dashv b ) \dashv c, \label{di-1}\\
a \dashv (b \vdash c) =~& (a \dashv b ) \dashv c, \label{di-2}\\
a \vdash (b \dashv c) =~& (a \vdash b) \dashv c, \label{di-3}\\
a \vdash (b \vdash c) =~& (a \dashv b) \vdash c, \label{di-4}\\
a \vdash (b \vdash c) =~& (a \vdash b) \vdash c \label{di-5}.
\end{align}
Diassociative algebras are related to Leibniz algebras in the same way associative algebras are related to Lie algebras \cite{loday-pirashvili}. In \cite{Frabetti,loday1}, the authors developed the cohomology theory of diassociative algebras. It has been observed that the combinatorics of planar binary trees play a crucial role in their construction of cohomology. In \cite{yau,maj-muk,Das}, the authors showed that the cohomology of a diassociative algebra also carries a Gerstenhaber algebra structure by constructing certain combinatorial operations on planar binary trees. The cohomology of a diassociative algebra is closely related to extensions and formal deformations of the structure \cite{majumdar-mukherjee}. Averaging operators on associative algebras are closely related to diassociative algebras \cite{loday1}. In \cite{das-cont}, the author extends this result further to the cohomological set-up. See also \cite{Frabetti,loday1, majumdar-mukherjee} for more studies on diassociative algebras. The Koszul dual of the operad of diassociative algebras is the operad of dendriform algebras \cite{loday1}. They are equally important due to their close connections with Rota-Baxter operators, splitting of operations, pre-Lie algebras and the Yang-Baxter equation (see \cite{bai-guo-ni} and the references therein). 
   

\medskip

In this paper, we introduce and study diassociative conformal algebras. By definition, a diassociative conformal algebra unifies associative conformal algebras on the one hand and diassociative algebras on the other. We show that a diassociative conformal algebra gives rise to a Leibniz conformal algebra in the same way a diassociative algebra produces a Leibniz algebra (cf. Proposition \ref{prop-di-leib}). We show that diassociative conformal algebras appear naturally in the study of formal distribution diassociative algebras. More precisely, the category of diassociative conformal algebras is equivalent to the category of equivalence classes of formal distribution diassociative algebras (cf. Theorem \ref{thm-categ-equiv}). Next, given a $\C[\p]$-module $\D$, we construct a nonsymmetric operad $\mathcal{O}$ that yields a graded Lie algebra whose Maurer-Cartan elements correspond to diassociative conformal algebra structures on $\D$ (cf. Theorem \ref{thm-operad}, Proposition \ref{1-1}). Then we define the cohomology of a diassociative conformal algebra and show that the cohomology inherits a Gerstenhaber algebra structure (cf. Theorem \ref{thm-gers}). We show that the cohomology of a diassociative conformal algebra is isomorphic to the cohomology of a suitable subcomplex of its annihilation diassociative algebra (cf. Theorem \ref{th-39}). As a consequence, we describe the cohomology of a current diassociative conformal algebra in terms of a suitable polynomial diassociative algebra (cf. Proposition \ref{prop-cur}). Finally, we will use our cohomology to describe extensions and deformations of the diassociative conformal structure (cf. Theorem \ref{thm-ext}, Theorem \ref{thm-def}).

\medskip

The notion of averaging operators on Lie conformal algebras was considered in \cite{Koles} while studying the singular part of a solution of the classical Yang-Baxter equation on a Lie algebra. In this paper, we consider averaging operators on associative conformal algebras. A more general notion could be a relative averaging operator on an associative conformal algebra with respect to a conformal bimodule. We observe that a (relative) averaging operator on an associative conformal algebra induces a diassociative conformal algebra structure (cf. Proposition \ref{prop*}). Conversely, given any diassociative conformal algebra, one can construct a relative averaging operator so that the induced diassociative conformal structure is the prescribed one (cf. Theorem \ref{thm-diass-to-avg}, Proposition \ref{prop-left}). Motivated by the theory of averaging operators in the classical setting \cite{das-cont}, we find a Maurer-Cartan characterization and define cohomology of relative averaging operators on an associative conformal algebra (cf. Theorem \ref{thm-mc-avg}).

\medskip

The paper is organized as follows. In Section \ref{sec2}, we introduce and study diassociative conformal algebras. Among other results, we show that the category of diassociative conformal algebras is equivalent to the category of equivalence classes of formal distribution diassociative algebras. Section \ref{sec3} is devoted to the Maurer-Cartan characterization and cohomology of diassociative conformal algebras. We show that the cohomology carries a natural Gerstenhaber algebra structure. In Section \ref{sec4}, we study extensions and deformations of diassociative conformal algebras using the cohomology theory developed earlier. Finally, Section \ref{sec5} considers (relative) averaging operators, their relations with diassociative conformal algebras and cohomology associated with such operators.

\medskip

\section{Diassociative Conformal Algebras}\label{sec2}
In this section, we first introduce diassociative conformal algebras and provide various examples. We establish some useful identities in a diassociative conformal algebra and show that a suitable skew-symmetrization of this structure yields a Leibniz conformal algebra. We end this section by showing that the category of diassociative conformal algebras is equivalent to the category of equivalence classes of formal distribution diassociative algebras.

\begin{definition}
 A \textbf{diassociative conformal algebra} (or simply a {\bf diconformal algebra}) is a $\C[\p]$-module $\D$ equipped with two $\C$-bilinear maps (called $\la$-products) $\dashv_{\la}, ~  \! \vdash_{\la} : \D \times \D \rightarrow \D[ \la ]$ both of which are conformal sesquilinear in the sense that
 \begin{align}
     \partial a \dashv_\la b = - \la (a \dashv_\la b), \quad a \dashv_\la \partial b = (\partial + \la) ~  \! a \dashv_\la b, \label{ses1}\\
     \partial a \vdash_\la b = - \la (a \vdash_\la b), \quad a \vdash_\la \partial b = (\partial + \la) ~  \! a \vdash_\la b, \label{ses2}
 \end{align}
 that satisfy the following list of identities: 
    \begin{align}
   a\dashv_{\la}(b \dashv_{\mu}c)=~&(a\dashv_{\la}b)\dashv_{\la + \mu}c,\label{equation 1} \\
 a \dashv_{\la}(b \vdash_{\mu}c) =~& (a\dashv_{\la}b)\dashv_{\la + \mu}c,\label{equation 2}\\
  a\vdash_{\la}(b \dashv_{\mu}c) =~&  (a\vdash_{\la}b)\dashv_{\la + \mu}c,\label{equation 3}\\
   a \vdash_{\la}(b\vdash_{\mu}c)=~&(a\dashv_{\la}b)\vdash_{\la+\mu}c,\label{equation 4}\\
  a \vdash_{\la}(b\vdash_{\mu}c)=~&(a\vdash_{\la}b)\vdash_{\la + \mu}c, \label{equation 5}
\end{align} 
for all $a,b,c \in \D$. A diassociative conformal algebra as above may be denoted by the triple $(\D, \dashv_{\la},\vdash_{\la})$ or simply by $\D$.
\end{definition}

It follows from identities (\ref{equation 1}) and (\ref{equation 5}) that each of the $\la$-products $\dashv_{\la}$ and $\vdash_{\la}$ in a diassociative conformal algebra independently satisfies the conformal associativity. Furthermore, these two $\la$-products are related by the additional conformal associativity-type identities (\ref{equation 2}), (\ref{equation 3}), and (\ref{equation 4}). Consequently, every associative conformal algebra $(\mathcal{A}, \cdot_{\la} \cdot)$ can be viewed naturally as a diassociative conformal algebra by identifying both the operations $\dashv_{\la}$ and $\vdash_{\la}$ with the associative $\la$-product $\cdot_{\la} \cdot$ on $\mathcal{A}$.

\medskip

Let $(\D,\dashv_{\la},\vdash_{\la})$ be a diassociative conformal algebra. For any $j \geq 0$, we define the $j$-th products for the associative conformal algebras $(\D, \dashv_{\la})$ and $(\D, \vdash_{\la})$ by taking  
\begin{align} \label{j}
    a \dashv_{\la} b = \sum_{j  \geq 0} ( a \dashv_{(j)} b ) ~  \! \frac{\la^j}{j!} ~~~~~ \text{ and } ~~~~~  a \vdash_{\la}b = \sum_{j \geq 0} ( a\vdash_{(j)}b ) ~ \! \frac{\la^j}{j!}.
\end{align}
Then the conformal sesquilinearity conditions (\ref{ses1}) and (\ref{ses2}) can be understood as
\begin{align*}
    \partial a \dashv_{(j)} b = - j ~ \! (a \dashv_{(j-1)} b), \quad a \dashv_{(j)} \partial b = j ~ \! ( a \dashv_{(j-1)} b) + \partial (a \dashv_{(j)} b),\\
     \partial a \vdash_{(j)} b = - j ~ \! (a \vdash_{(j-1)} b), \quad a \vdash_{(j)} \partial b = j ~ \! ( a \vdash_{(j-1)} b) + \partial (a \vdash_{(j)} b).
\end{align*}
Moreover, the diassociative conformal algebra identities (\ref{equation 1})--(\ref{equation 5}) are equivalent to the following identities expressed in terms of the $j$-th products:
\begin{align}
  a\dashv_{(m)}(b \dashv_{(n)}c) =~& \sum_{r = 0}^{m} \binom{m}{r}(a\dashv_{(r)}b) \dashv_{(m+n-r)}c,\label{equation 6}\\
 a \dashv_{(m)}(b \vdash_{(n)}c) =~& \sum_{r = 0}^{m} \binom{m}{r}(a\dashv_{(r)}b) \dashv_{(m+n-r)}c,\label{equation 7}\\
 a\vdash_{(m)}(b \dashv_{(n)}c) = ~&  \sum_{r = 0}^{m} \binom{m}{r}(a\vdash_{(r)}b)\dashv_{(m+n-r)}c,\label{equation 8}\\
 a \vdash_{(m)}(b\vdash_{(n)}c)=~&\sum_{r = 0}^{m} \binom{m}{r} (a\dashv_{(r)}b)\vdash_{(m+n-r)}c,\label{equation 9}\\
 a \vdash_{(m)}(b\vdash_{(n)}c)=~&\sum_{r = 0}^{m} \binom{m}{r} (a\vdash_{(r)}b)\vdash_{(m+n-r)}c,\label{equation 10}
\end{align}
for all $a,b,c\in\D$ and $m,n\geq 0$.

\begin{definition}
Let $(\D,\dashv_{\la},\vdash_{\la})$ and
$(\D',\dashv_{\la}',\vdash_{\la}')$ be two diassociative conformal algebras. Then a {\bf homomorphism} of diassociative conformal algebras
from $\D$ to
$\D'$
is a $\C[\p]$-linear map
$\varphi:\D\rightarrow\D'$ satisfying
\begin{align*}
\varphi(a\dashv_{\la}b)
=\varphi(a)\dashv_{\la}'\varphi(b) ~~~~ \text{ and } ~~~~
\varphi(a\vdash_{\la}b)
=\varphi(a)\vdash_{\la}'\varphi(b), \text{ for all } a, b \in \D.
\end{align*}
\end{definition}
\noindent We denote by $\mathbf{DiassC}$ the category whose objects are diassociative conformal algebras and morphisms are homomorphisms of diassociative conformal algebras.

\medskip

Let $(\D,\dashv_{\la},\vdash_{\la})$ be a diassociative conformal algebra. Then a {\bf subalgebra} of $\D$ is a $\C[\p]$-submodule $\E\subseteq\D$ that is closed under both the $\la$-products, i.e., $a\dashv_{\la}b \in\E[\la]$ and $a\vdash_{\la}b\in\E[\la]$, for all $a,b\in\E$. On the other hand, an {\bf ideal} of $\D$ is a $\C[\p]$-submodule $\mathcal{I} \subseteq \D$ such that $a\dashv_{\la}b \in \mathcal{I} [\la]$ and $ a\vdash_{\la}b \in\mathcal{I}[\la]$, whenever $a \in \mathcal{I}$ or $b \in \mathcal{I}$.


\begin{proposition}
Let $(\mathcal{D}, \dashv_\la, \vdash_\la)$ be a diassociative conformal algebra, where $\mathcal{D}=\mathbb{C}[\partial]a$ is a rank one free $\mathbb{C}[\partial]$-module. Then there exists a unique $\alpha \in \mathbb{C}$ such that
\begin{align*}
a\dashv_{\la}a=\alpha a
~~~~\text{ and }~~~~~
a\vdash_{\la}a=\alpha a.
\end{align*}
\end{proposition}

\begin{proof}
Since $\mathcal{D}=\mathbb{C}[\partial]a$, there exist polynomials
$f(\partial,\la), \, g(\partial,\la)\in\mathbb{C}[\partial,\la]$ such that $a\dashv_{\la}a=f(\partial,\la)a$ and $a\vdash_{\la}a=g(\partial,\la)a$.
Replacing $\mu$ by $-\la$ in the diassociative conformal identity (\ref{equation 1}), while keeping $\la$ unchanged, yields
\begin{align}\label{p-eq}
f(\partial+\la,-\la)f(\partial,\la)
=f(0,\la)f(\partial,0).
\end{align}
If $f=0$, then $f$ is constant with value $\alpha=0$. Suppose now that $f\neq0$, and write $f(\partial,\lambda)=\sum \limits_{i=0}^{k}a_{i}(\la)\p^i$, where $a_i(\la) \in \C[\la]$ and $a_k(\lambda)\neq0$. For $k>0$, comparing the coefficients of $\partial^{2k}$ on both sides of (\ref{p-eq}), we obtain $a_k(\la)a_k(-\la)=0$. This is not possible as $\mathbb{C}[\la]$ is an integral domain and $a_k(\la)$ is a non-zero polynomial. Hence $k=0$, and consequently $f(\partial,\la)=a_0(\la)$. Substituting this expression into (\ref{p-eq}) gives
\begin{align}\label{eqn-lambda}
a_0(-\la)a_0(\la)=a_0(\la)a_0(0).
\end{align}
Next, let $a_0(\la)= \sum \limits_{i=0}^{t}b_i \la^i$, where $b_i \in \C$ and $b_t\neq0$. If $t>0$, then by comparing the coefficients of $\la^{2t}$ on both sides of (\ref{eqn-lambda}) gives $(-1)^t b_t^2=0$, which is not possible as $b_t\neq0$. Therefore, $t=0$ and hence $f(\partial,\la)=\alpha$, for some $\alpha\in\mathbb{C}$. Similarly, from the diassociative conformal identity (\ref{equation 5}), we obtain $g(\partial,\la)=\beta$, for some $\beta\in\mathbb{C}$. Consequently, $a\dashv_{\la}a=\alpha  a$ and $a\vdash_{\la}a=\beta a$.
Finally, using the identities (\ref{equation 2}) and (\ref{equation 4}), we obtain $\alpha^2=\alpha\beta$ and $\beta^2=\alpha\beta$. This gives $\alpha=\beta$, which proves the result.
\end{proof}

The above proposition shows that any diassociative conformal algebra structure on a rank one free $\mathbb{C}[\partial]$-module is determined by a complex number $\alpha \in \mathbb{C}$. It turns out that up to isomorphism, there are exactly two diassociative conformal algebras on a rank one free $\mathbb{C}[\partial]$-module corresponding to $\alpha = 0$ and $\alpha =1$, respectively.

\begin{exam}\label{cur-d}
Let $(D,\dashv,\vdash)$ be a diassociative algebra. Consider the $\C[\partial]$-module $\operatorname{Cur}(D)=\C[\partial]\otimes D$. We define two $\la$-products on $\operatorname{Cur}(D)$ by
\begin{align*}
\big(p(\partial)\otimes a\big)\dashv_{\la}\big(q(\partial)\otimes b\big)
&:=
\big(p(-\la)q(\la+\partial)\big)\otimes(a\dashv b),\\
\big(p(\partial)\otimes a\big)\vdash_{\la}\big(q(\partial)\otimes b\big)
&:=
\big(p(-\la)q(\la+\partial)\big)\otimes(a\vdash b),
\end{align*}
for all $p(\partial),q(\partial)\in\C[\partial]$ and $a,b\in D$. Then $(\operatorname{Cur}(D),\dashv_{\la},\vdash_{\la} )$ is a diassociative conformal algebra. This is called the
\emph{current diassociative conformal algebra} associated with the
diassociative algebra $D$.
\end{exam}

\begin{exam}
Let $\mathcal{A}$ be an associative conformal algebra and $d:\mathcal{A}\to\mathcal{A}$ be a $\C[\partial]$-linear map satisfying
\begin{align*}
d^{2}=0 ~~~ \text{ and } ~~~ d(a_{\la}b)=(da)_{\la}b+a_{\la}(db), \text{ for all }a,b\in\mathcal{A}.
\end{align*}
In this case, the pair $(\mathcal{A},d)$ is often called a differential associative conformal algebra. We define two $\la$-products on $\mathcal{A}$ by
\begin{align*}
a\dashv_{\la}b:=a_{\la}(db) ~~~ \text{ and } ~~~ 
a\vdash_{\la}b:=(da)_{\la}b, \text{ for }a,b\in\mathcal{A}.
\end{align*}
Then $(\mathcal{A},\dashv_{\la},\vdash_{\la})$ is a diassociative conformal algebra.
\end{exam}

\begin{exam}
Let $\mathcal{A}$ be an associative conformal algebra. A {\bf conformal $\mathcal{A}$-bimodule} is a $\C [\partial]$-module $\M$ equipped with two conformal sesquilinear maps (called left and right $\la$-actions) $ \mathcal{A} \times \M \rightarrow \M [\la]$, $(a, u) \mapsto a_\la u$ and $\M \times \mathcal{A} \rightarrow \M [\la]$, $(u, a) \mapsto u_\la a$ satisfying
\begin{align*}
a_{\la}(b_{\mu}u)=(a_{\la}b)_{\la+\mu}u, \quad
a_{\la}(u_{\mu}b)=(a_{\la}u)_{\la+\mu}b ~~~ \text{ and } ~~~
u_{\la}(a_{\mu}b)=(u_{\la}a)_{\la+\mu}b, \text{ for }a,b\in\mathcal{A},\, u \in\mathcal{M}.
\end{align*}
Given a conformal $\mathcal{A}$-bimodule $\M$, let $f:\M\to\mathcal{A}$ be a conformal $\mathcal{A}$-bimodule homomorphism. That is, $f$ is a $\C[\partial]$-linear map satisfying $f(a_{\la}u)=a_{\la}f(u)$ and $f(u_{\la}a)=f(u)_{\la}a$, for all $a \in \mathcal{A}$, $u \in \M$. Define two $\la$-products on $\M$ by
\begin{align*}
u\dashv_{\la}v:=u_{\la}f(v) ~~~ \text{ and } ~~~ 
u\vdash_{\la}v:=f(u)_{\la}v, \text{ for all } u,v\in\M.
\end{align*}
Then $(\M,\dashv_{\la},\vdash_{\la})$ is a diassociative conformal algebra.
\end{exam}
\begin{exam}
Let $\mathcal{A}$ be an associative conformal algebra, and take $\mathcal{D}=\underbrace{\mathcal{A}\oplus\cdots\oplus\mathcal{A}}_{n\text{ copies}}$ be the direct sum of $n$ copies of $\mathcal{A}$. We define two $\la$-products on $\mathcal{D}$ by
\begin{align*}
(a_1,\ldots,a_n)\dashv_{\la}(b_1,\ldots,b_n)
&:=
\big(
a_{1\la} (\sum_{i=1}^{n}b_i ),
\ldots,
a_{n\la} (\sum_{i=1}^{n}b_i )
\big),\\
(a_1,\ldots,a_n)\vdash_{\la}(b_1,\ldots,b_n)
&:=
\big(
(\sum_{i=1}^{n}a_i)_{\la}b_1,
\ldots,
(\sum_{i=1}^{n}a_i )_{\la}b_n
\big),
\end{align*}
for $(a_1,\ldots,a_n),(b_1,\ldots,b_n)\in\mathcal{D}$. Then $(\mathcal{D},\dashv_{\la},\vdash_{\la})$ is a diassociative conformal algebra.
\end{exam}

In the following result, we show that the defining identities (\ref{ses1})-(\ref{equation 5}) of a diassociative conformal algebra imply some other useful identities. We will use this to show that a diassociative conformal algebra yields a Leibniz conformal algebra structure.

\begin{lemma} \label{ lemma 1}
    Let $(\D, \dashv_{\la}, \vdash_{\la})$ be a diassociative conformal algebra. Then for any $a,b,c \in \D$, the following identities hold: 
    \begin{align}
         a \dashv_{-\p -\la}(b \vdash _{\mu}c) =~& (a \dashv _{-\p -\mu} b) \dashv _{-\p + \mu-\la } c, \label{edi-1}\\
        a \dashv_{-\p -\la}( b \dashv _{-\p -\mu} c) =~& (a \dashv _{-\p+ \mu -\la } b) \dashv _{-\p -\mu} c, \label{edi-2}\\
        a \vdash_{-\p - \la}(b \vdash_{\mu} c) =~& (a \dashv_{-\p -\mu}b)\vdash_{-\p + \mu-\la }c,\label{edi-3}\\
        a \vdash _{\la} (b \dashv _{-\p -\mu} c) =~& (a \vdash _{\la} b) \dashv_{ -\p -\mu} c\label{edi-4}.
    \end{align}
\end{lemma}

\begin{proof}
    For any $a,b,c \in \D$, we observe that
\begin{align*}
  & a \dashv_{-\p -\la}(b \vdash_{\mu}c) \stackrel{(\ref{equation 2})}{=} (a \dashv_{-\p -\la}b) \dashv_{-\p+\mu-\la}c 
   \stackrel{(\ref{equation 1})}{=}a \dashv_{-\p -\la}(b \dashv_{\mu} c)
    \stackrel{(\ref{equation b})}{=} (a \dashv_{-\p -\mu}b) \dashv_{-\p +\mu-\la} c,\\
  & \qquad \qquad \qquad \qquad \qquad a \dashv_{-\p-\la}(b \dashv_{-\p-\mu}c) \stackrel{(\ref{equation c})}{=} (a\dashv_{-\p+\mu-\la}b)\dashv_{-\p-\mu}c,\\
  &   a \vdash_{-\p - \la}(b \vdash_{\mu} c)  \stackrel{(\ref{equation b})}{=}(a \vdash_{-\p-\mu}b)\vdash_{-\p+\mu-\la}c
    \stackrel{(\ref{equation 5})}{=}a \vdash_{-\p-\mu}(b \vdash_{-\la+2\mu}c)
     \stackrel{(\ref{equation 4})}{=}(a\dashv_{-\p-\mu}b)\vdash_{-\p+\mu-\la}c.
\end{align*}
This proves the first three identities. Finally, to prove the last identity, we consider the conformal sesquilinearity condition $a \vdash_{\la}\left( (-\p-\mu)^{j} b \right)=(-\p-\la-\mu)^{j} (a \vdash_{\la}b)$. Thus,

\begin{align*}
a \vdash_{\la}(b \dashv_{-\partial-\mu}c)
= a \vdash_\la \big(  \sum_{j \geq 0} \frac{(-\partial - \mu)^j}{j!} ~ \! b \dashv_{(j)} c  \big) =~& 
\sum_{j\geq 0}
\frac{(-\partial-\la-\mu)^j}{j!}
a \vdash_{\la}(b \dashv_{(j)}c) \\
=~&
a \vdash_{\la}(b \dashv_{\nu}c)\big|_{\nu=-\la-\mu-\partial} \\
=~& (a \vdash _{\la} b) \dashv_{ -\p -\mu} c \quad (\text{by } (\ref{equation 3})).
\end{align*}
This concludes the result.
\end{proof}

In \cite{loday1}, Loday showed that a diassociative algebra naturally gives rise to a Leibniz algebra, in the same way an associative algebra gives a Lie algebra via skew-symmetrization. To extend Loday's result to the conformal context, we first recall the following.

\begin{definition}
    (\cite{bakalov-kac-voronov,zhang}) A {\bf Leibniz conformal algebra} is a pair $(\mathcal{L},[\cdot_{\la}\cdot])$ consisting of a $\C[\p]$-module $\mathcal{L}$ with a conformal sesquilinear $\la$-bracket $[\cdot_{\la}\cdot]:\mathcal{L} \times \mathcal{L} \to \mathcal{L}[   \la ] $ that satisfies
\begin{align}\label{conf-leib}
 [a_{\la}[b_{\mu}c]] =[[a_{\la}b]_{\la+\mu}c]
+ [b_{\mu}[a_{\la}c]], \text{ for all } a,b,c \in \mathcal{L}. 
\end{align}
\end{definition}

A Leibniz conformal algebra for which the $\la$-bracket is skew-symmetric turns out to be a Lie conformal algebra. A homomorphism between two Leibniz conformal algebras is a $\C [\p]$-linear map preserving the $\la$-brackets. See \cite{das-sahoo, zhang,zhou} for further studies on Leibniz conformal algebras.

\begin{proposition}\label{prop-di-leib}
    Let $(\D,\dashv_{\la},\vdash_{\la})$ be a diassociative conformal algebra. Then we define a $\la$-bracket $[ \cdot_\la \cdot ] : \D \times \D \rightarrow \D [\la]$ by
\begin{align}\label{lambda-bracket}
    [a_{\la}b] := a\vdash_{\la}b ~ \! - ~ \! b\dashv_{-\la-\partial}a, \text{ for } a,b\in\D.
\end{align}
Then $(\D, [ \cdot_\la \cdot ])$ is a Leibniz conformal algebra. 
\end{proposition}

\begin{proof}
It is immediate that the $\la$-bracket $[\cdot_\la \cdot ]$ satisfies the conformal sesquilinearity conditions. Thus, it remains to verify the identity (\ref{conf-leib}). For any $a,b,c\in\D$, we observe that
\begin{align*}
&[[a_{\la}b]_{\la+\mu}c]+[b_{\mu}[a_{\la}c]]-[a_{\la}[b_{\mu}c]]\\
={}&[a_{\la}b]\vdash_{\la+\mu}c
-c\dashv_{-\la-\mu-\partial}[a_{\la}b]
+b\vdash_{\mu}[a_{\la}c]
-[a_{\la}c]\dashv_{-\mu-\partial}b\\
&\hspace{1cm}
-a\vdash_{\la}[b_{\mu}c]
+[b_{\mu}c]\dashv_{-\la-\partial}a\\
={}&(a\vdash_{\la}b-b\dashv_{-\la-\partial}a)\vdash_{\la+\mu}c
-c\dashv_{-\la-\mu-\partial}(a\vdash_{\la}b-b\dashv_{-\la-\partial}a)\\
&\hspace{1cm}
+b\vdash_{\mu}(a\vdash_{\la}c-c\dashv_{-\la-\partial}a)
-(a\vdash_{\la}c-c\dashv_{-\la-\partial}a)\dashv_{-\mu-\partial}b\\
&\hspace{1cm}
-a\vdash_{\la}(b\vdash_{\mu}c-c\dashv_{-\mu-\partial}b)
+(b\vdash_{\mu}c-c\dashv_{-\mu-\partial}b)\dashv_{-\la-\partial}a\\
={}&0 \quad (\text{by } (\ref{equation 1}), (\ref{equation 4}) \text{ and Lemma } \ref{ lemma 1}).
\end{align*} 
Hence, the conformal Leibniz identity holds.
\end{proof}

The above construction is functorial, yielding a functor $\mathcal{F}: {\bf DiassC} \rightarrow {\bf LeibC}$ from the category of diassociative conformal algebras to the category of Leibniz conformal algebras.

In a diassociative conformal algebra, if the $\la$-products $\dashv_\la$ and $\vdash_\la$ coincide (equivalently, it is given by an associative conformal algebra), then the corresponding $\la$-bracket (\ref{lambda-bracket}) is skew-symmetric. Hence, we obtain a Lie conformal algebra.

\medskip
   
  \noindent {\bf Equivalence of categories.} We construct diassociative conformal algebras from formal distribution diassociative algebras and prove that the category of diassociative conformal algebras is equivalent to the category of equivalence classes of formal distribution diassociative algebras. 
 
 Let $( D,\dashv, ~\!\vdash)$ be a diassociative algebra. A $D${\bf -valued formal distribution} in one indeterminate $w$ is a formal series of the form $a(w)=\sum \limits_{m \in \mathbb{Z}} a_{m} w^{-m-1}$, where the coefficients $a_{m} \in D$ are called the Fourier modes of the distribution $a(w)$. The set of all $D$-valued formal distributions forms a complex vector space, denoted by $D\llbracket w^{\pm1}\rrbracket$.
 Next, we introduce two $\la$-products and the corresponding $j$-th products on $D\llbracket w^{\pm1}\rrbracket$.
 \begin{definition}
    Let $D$ be a diassociative algebra. Then the $\la$-products on $D \llbracket w^{\pm 1 }\rrbracket$ are $\C$-bilinear maps $ \dashv_{\la}, \vdash_{\la}: D \llbracket w^{\pm 1 }\rrbracket \times D \llbracket w^{\pm 1 }\rrbracket \rightarrow D \llbracket w^{\pm 1 }\rrbracket  \llbracket  \la  \rrbracket$, defined by 
    \begin{align}\label{la-bra}
    a(w) \dashv_{\la} b(w)  = \text{Res}_{z} e^{\la(z-w)} a(z) \dashv b(w)~ \text{ and } ~a(w) \vdash_{\la} b(w) = \text{Res}_{z}e^{\la(z-w)}a(z) \vdash b(w),
\end{align}
  where
    \begin{align*}
& \qquad \qquad \qquad a(w)=\sum \limits_{m \in \mathbb{Z}} a_{m} w^{-m-1}, \quad b(w)=\sum \limits_{n \in \mathbb{Z}} b_{n} w^{-n-1} \in D\llbracket w^{\pm1}\rrbracket, \\
&a(z) \dashv b(w)= \sum_{m,n \in \mathbb{Z} }(a_m \dashv b_n) ~ \! z^{-m-1}w^{-n-1} ~  \text{ and } ~   a(z) \vdash b(w)= \sum_{m,n \in \mathbb{Z} }(a_m \vdash b_n) ~ \! z^{-m-1}w^{-n-1}.
\end{align*}
\end{definition}
\noindent Thus, $a(w)\dashv_{\la}b(w)$ and $a(w)\vdash_{\la}b(w)$ are the formal Fourier transforms, with respect to $z$, of $a(z)\dashv b(w)$ and $a(z)\vdash b(w)$, respectively. If we expand the right-hand side of the $\la$-products defined above, we get
\begin{align*}
    &a(w) \dashv_{\la} b(w)  = \sum_{j \geq 0} \big( \text{Res}_{z} (z-w)^{j} a(z) \dashv b(w) \big) \frac{\la^j}{j!},\\
    &a(w) \dashv_{\la} b(w)  = \sum_{j \geq 0} \big( \text{Res}_{z} (z-w)^{j} a(z) \dashv b(w) \big) \frac{\la^j}{j!}.
\end{align*}
These expressions of $\la$-products motivate us to define the $j$-th products between two $D$-valued formal distributions.

 \begin{definition} \label{D-j}
 Let $D$ be a diassociative algebra and $j$ be a non-negative integer. The {\bf $j$-th products} on $D\llbracket w^{\pm1}\rrbracket$ are $\C$-bilinear maps $\dashv_{(j)}, \vdash_{(j)}: D \llbracket w^{\pm 1 }\rrbracket \times D \llbracket w^{\pm 1 } \rrbracket \rightarrow D \llbracket w^{\pm 1 }\rrbracket$, defined by 
    \begin{align}\label{j-th}
      a(w) \dashv_{(j)} b(w)  = \text{Res}_{z}(z-w)^{j}a(z) \dashv b(w) ~~~\text{ and } ~~~  a(w) \vdash_{(j)} b(w) = \text{Res}_{z}(z-w)^{j}a(z) \vdash b(w),
    \end{align}
    for all $a(w),b(w) \in D\llbracket w^{\pm1}\rrbracket$.
 \end{definition}
 
We equip the space $D\llbracket w^{\pm1}\rrbracket$ with the structure of a $\C[\partial]$-module, where the formal differentiation operator $\partial_w$ gives the $\partial$-action. That is, $\partial (a (w)) =\p_wa(w)$, for $a(w)\in D\llbracket w^{\pm1}\rrbracket$. Then the following straightforward proposition describes the interaction between the $\partial$-action and the $\la$-products.
\begin{proposition}\label{c.s.lap}
   Let $D$ be a diassociative algebra and $a(w), b(w) \in D\llbracket w^{\pm1}\rrbracket$.  Then, the $\la$-products defined in (\ref{la-bra}) satisfy the following identities:
    \begin{align*}
        \p a(w) \dashv_{\la} b(w) &= - \la ~ \! a(w) \dashv_{\la} b(w), \qquad  a(w) \dashv_{\la} \p  b(w) = (\p + \la) ~ \! a(w) \dashv_{\la} b(w), \\
        \p a(w) \vdash_{\la} b(w) &= - \la ~ \! a(w) \vdash_{\la} b(w), \qquad  a(w) \vdash_{\la} \p  b(w) = (\p + \la) ~ \!  a(w) \vdash_{\la} b(w).
    \end{align*}
\end{proposition}
    

This shows that the $\la$-products of $D$-valued formal distributions satisfy the conformal sesquilinearity conditions. Moreover, translating the defining identities of the diassociative algebra $D$ into identities for the $\la$-products and the corresponding $j$-th products of $D$-valued formal distributions, we get the following results.
\begin{proposition}
    Let $D$ be a diassociative algebra. Then the $\la$-products of $D$-valued formal distributions satisfy the identities {\rm (\ref{equation 1})--(\ref{equation 5})}.
\end{proposition}
\begin{proposition}
    Let $D$ be a diassociative algebra. Then the $j$-th products (\ref{j-th}) of $D$-valued formal distributions satisfy the identities {\rm (\ref{equation 6})--(\ref{equation 10})}.
\end{proposition}
We will now define the notion of locality of a pair of $D$-valued formal distributions.
\begin{definition}
    Let $D$ be a diassociative algebra. A pair $\big(a(w),b(w)\big)$ of $D$-valued formal distributions is called a {\bf local pair} if there exists a non-negative integer $N$ such that 
    \begin{align*}
        (z-w)^{N} ~ \! a(z) \dashv b(w) =0 ~~~\text{ and } ~~~   (z-w)^{N} ~ \! a(z) \vdash b(w) =0.
    \end{align*}
\end{definition}
The locality condition has an immediate consequence for the $j$-th products. Indeed, if $\big(a(w),b(w)\big)$ is a local pair of $D$-valued formal distributions, then there exists a non-negative integer $N$ such that  
\begin{align*}
    a(w) \dashv_{(j)} b(w) =0 ~  \text{ and } ~ a(w) \vdash_{(j)}b(w)=0, \text{ for all } j \geq N.
    \end{align*}
   Consequently, only finitely many $j$-th products are nonzero. It follows that the corresponding $\la$-products are polynomials in $\la$ with coefficients in $D\llbracket w^{\pm1}\rrbracket$, that is, they belong to $D\llbracket w^{\pm1}\rrbracket[\la]$ rather than the larger space $D\llbracket w^{\pm1}\rrbracket\llbracket\la\rrbracket$. Thus, the locality condition is precisely what ensures the polynomiality of the $\la$-products.

Let $\mathcal{F}$ be a family of $D$-valued formal distributions such that every pair of elements of $\mathcal{F}$ forms a local pair and the coefficients of all distributions in $\mathcal{F}$ span $D$. The pair $(D,\mathcal{F})$ is called a {\bf formal distribution diassociative algebra}. Let $\overline{\mathcal{F}}$ denote the smallest subspace of $D \llbracket w^{\pm 1 }\rrbracket$ containing $\mathcal{F}$ which is closed under all $j$-th products (\ref{j-th}) and $\p$-invariant. Then we have the following result.
\begin{proposition}
    Let $(D,\mathcal{F})$ be a formal distribution diassociative algebra. Then $\overline{\mathcal{F}}$ is a local family.
\end{proposition}
\begin{proof}
    Let $a(w), b(w) \text{ and } c(w)$ be three elements of $\mathcal{F}$. Then it is sufficient to prove that the pairs $\big( a(w)\dashv_{(j)}b(w),\,c(w) \big)$ and $\big( a(w)\vdash_{(j)}b(w),\,c(w) \big)$ are both local pairs for every $j \geq 0$. This follows immediately from Dong's lemma for diassociative algebras \cite{koli-Sart}.
\end{proof}
Then one may easily verify that the $\la$-products $\dashv_\la, ~ \! \vdash_\la: \overline{\mathcal{F}} \times \overline{\mathcal{F}} \to \overline{\mathcal{F}}[\la]$ defined by 
\begin{align*}
    a(w) \dashv_\la b(w) = \sum_{j \geq 0} a(w) \dashv_{(j)} b(w) ~ \! \frac{\la^j}{j!} ~\text{ and }~  a(w) \vdash_\la b(w) = \sum_{j \geq 0} a(w) \vdash_{(j)} b(w) ~ \! \frac{\la^j}{j!}, 
\end{align*}
for $a(w),b(w) \in \overline{\mathcal{F}}$, makes $\overline{\mathcal{F}}$ into a diassociative conformal algebra. We denote it by $\text{Conf}(D, \mathcal{F})$.

\medskip

Conversely, given a diassociative conformal algebra $(\D, \dashv_\la, \vdash_\la)$, we will construct a formal distribution diassociative algebra $( \mathrm{Diass} (\D), \mathcal{F})$. For this, we consider the $\C$-vector space $ \D[t^{\pm 1}]$ of all Laurent polynomials with coefficients in $\D$. We define two bilinear operations $\dashv$ and $\vdash$ on $\D[t^{\pm1}]$ by specifying their values on the generating elements as:
\begin{align} \label{equation 11} 
    at^{m} \dashv bt^{n} := \sum_{j \geq 0 } \binom{m}{j}  \big(a \dashv_{(j)} b \big) ~ \! t^{m+n-j }~ \text{ and }~  at^{m} \vdash bt^{n} := \sum_{j \geq 0 } \binom{m}{j} \big(a \vdash_{(j)} b \big) ~ \! t^{m+n-j }, 
\end{align}
for $ a, b \in \D$ and $ m,n \in \mathbb{Z} $. Here the operations $\dashv_{(j)}$ and $\vdash_{(j)}$ are the $j$-th products corresponding to the $\la$-products $\dashv_\la$ and $\vdash_\la$, respectively. Moreover, $\binom{m}{j}$ is the extended binomial coefficient defined by
\begin{align*}
    \binom{m}{j} = ~
    \begin{cases}
        \frac{m(m-1) \cdots (m-j+1)}{j!}, & \text{ if } j > 0,\\
        1, & \text { if } j= 0.
    \end{cases}
\end{align*}
We note that the extended binomial coefficients satisfy the following two identities, which are essential for our purpose:
\begin{align}
    \binom{m}{l} \binom{l}{r} =~& \binom{m}{r} \binom{m-r}{l-r}, \text{ for }m \in \Z,\ l,r \in \mathbb{Z}_{\geq 0},\ 0 \leq r \leq l \label{f-1}\\
    \sum_{j=0}^{p}
\binom{n}{j}
\binom{m-r}{p-j}
=~&
\binom{m+n-r}{p},\text{ for }m,n \in \Z, \ r, p \in \mathbb{Z}_{\geq 0}. \label{f-2}
\end{align}
\begin{proposition}\label{di-al}
    The vector space $\D[t^{\pm 1} ]$ equipped with the bilinear operations $\dashv\text{ and }\vdash $ defined in {\rm (\ref{equation 11})} is a diassociative algebra.
\end{proposition}
\begin{proof}
   We will prove the diassociative identities (\ref{di-1})--(\ref{di-5}) only for the generating elements of $\D[t^{\pm 1} ]$. For any $a,b,c\in\D$ and $m,n,k\in\mathbb{Z}$, we have
    \begin{align*}
         at^{m} \dashv (bt^{n} \dashv ct^{k})  &~= \sum_{l , j \geq 0} \binom{m}{l} \binom{n}{j} ~ \! a\dashv_{(l)} (b \dashv_{(j)} c) ~ \! t^{m+n+k-l-j}\\&
         ~= \sum_{l , j \geq 0 } \binom{m}{l} \binom{n}{j} ~ \!  a\dashv_{(l)} (b \vdash_{(j)}c) ~ \! t^{m+n+k-j-l} \quad \left(\text{by } (\ref{equation 6}) \text{ and } (\ref{equation 7}) \right)\\
         &~=at^{m}\dashv(bt^{n}\vdash ct^{k}).
    \end{align*}
Also,
\begin{align*}
at^{m}\dashv(bt^{n}\dashv ct^{k})
&=\sum_{l,j \geq 0}
\binom{m}{l}\binom{n}{j}
~ \! a\dashv_{(l)}(b\dashv_{(j)}c) ~ \! 
t^{m+n+k-l-j}\\
&=\sum_{l,j \geq 0 }\sum_{r=0}^{l}
\binom{m}{l}\binom{n}{j}\binom{l}{r}
(a\dashv_{(r)}b)
\dashv_{(l+j-r)}c
 ~ \! t^{m+n+k-l-j} \quad  \left(\text{by } (\ref{equation 6})  \right)\\
&=\sum_{l,j \geq 0 }\sum_{r=0}^{l}
\binom{m}{r}\binom{m-r}{l-r}\binom{n}{j}
(a\dashv_{(r)}b)
\dashv_{(l+j-r)}c ~ \!
t^{m+n+k-l-j} \quad   \left(\text{by } (\ref{f-1})  \right)
\end{align*}
\begin{align*}
&=\sum_{p,r \geq 0 }
\binom{m}{r}
\sum_{j=0}^{p}
\binom{n}{j}
\binom{m-r}{p-j}
(a\dashv_{(r)}b)
\dashv_{(p)}c ~ \!
t^{m+n+k-p-r} \quad \left( \text{taking } p=l+j-r \right)\\
&=\sum_{p,r \geq 0 }
\binom{m}{r}
\binom{m+n-r}{p}
(a\dashv_{(r)}b)
\dashv_{(p)}c ~ \! 
t^{m+n+k-p-r} \quad \left(\text{by } (\ref{f-2})  \right)\\
&=(at^{m}\dashv bt^{n})\dashv ct^{k}.
\end{align*}
This verifies the identities (\ref{di-1}) and (\ref{di-2}). The remaining three identities (\ref{di-3})--(\ref{di-5}) can be verified analogously. 
This proves the result.
\end{proof}
Next, we consider the subspace $(\partial+\partial_t) ~ \! \D[t^{\pm1}]$ of the space $\D[t^{\pm1}]$, generated by the elements 
\begin{align*}
    (\partial a)t^m + mat^{m-1}, \text{ for } a\in\D, m\in \mathbb{Z}.
\end{align*}
Then we have the following result.
\begin{proposition}
The subspace $(\partial+\partial_t) ~ \! \D[t^{\pm1}]$ is a two-sided ideal of the diassociative algebra $\D[t^{\pm1}]$.
\end{proposition}
\begin{proof}
    For any $a,b \in \D$ and $m,n \in \Z$, we have
    \begin{align*}
      \big(  (\p a)t^{m} + mat^{m-1} \big) \dashv bt^{n} &= \sum_{j \geq 0 } \binom{m}{j}  \big(\p a \dashv_{(j)} b \big) ~ \! t^{m+n-j } + m \sum_{k \geq 0 } \binom{m-1}{k}  \big( a \dashv_{(k)} b \big) ~ \! t^{m+n-k-1 }\\
      &= -\sum_{j \geq 1 } j \binom{m}{j}  \big( a \dashv_{(j-1)} b \big) ~ \!  t^{m+n-j } +  \sum_{k \geq 0 }(m-k) \binom{m}{k}  \big( a \dashv_{(k)} b \big) ~ \! t^{m+n-k-1 }\\
      &= -\sum_{j \geq 0 } (m-j) \binom{m}{j}  \big( a \dashv_{(j)} b \big) ~ \! t^{m+n-j-1 } +  \sum_{k \geq 0 }(m-k) \binom{m}{k}  \big( a \dashv_{(k)} b \big) ~ \! t^{m+n-k-1 }\\
      &=0
    \end{align*}
    and 
    \begin{align*}
        & at^m \dashv \big( (\p b)t^n + nbt^{n-1} \big) \\
        =&~ \sum_{j \geq 0 } \binom{m}{j}  \big( a \dashv_{(j)} \p b \big) ~ \! t^{m+n-j } + n \sum_{k \geq 0 } \binom{m}{k}  \big( a \dashv_{(k)} b \big) ~ \! t^{m+n-k-1 } \\
        =&~\sum_{j \geq 0 } \binom{m}{j}  \big( \p( a \dashv_{(j)}  b)+ j( a \dashv_{(j-1)} b) \big) ~ \! t^{m+n-j } + n \sum_{k \geq 0 } \binom{m}{k}  \big( a \dashv_{(k)} b \big) ~ \! t^{m+n-k-1 }
        \end{align*}
        \begin{align*}
        =&~\sum_{j \geq 0 } \binom{m}{j}  ~ \! \p \big(a \dashv_{(j)}  b \big) ~ \! t^{m+n-j } +\sum_{j \geq 0 } \binom{m}{j}   (m+n-j) \big(a \dashv_{(j)}  b \big) ~ \! t^{m+n-j -1}\\
        &~ -\sum_{j \geq 0 } \binom{m}{j}   (m+n-j) \big(a \dashv_{(j)}  b \big) ~ \! t^{m+n-j-1 } +\sum_{j \geq 0 } j\binom{m}{j}   \big(  a \dashv_{(j-1)} b \big) ~ \! t^{m+n-j }\\
        &~+ n \sum_{k \geq 0 } \binom{m}{k}  \big( a \dashv_{(k)} b \big) ~ \! t^{m+n-k-1 }\\
        =&~ \alpha  - \sum_{j \geq 0 } \binom{m}{j}   (m-j) \big(a \dashv_{(j)}  b \big) ~ \! t^{m+n-j-1 } +\sum_{j \geq 0 } j\binom{m}{j}   \big(  a \dashv_{(j-1)} b \big) ~ \! t^{m+n-j }\\
        =&~\alpha  - \sum_{j \geq 0 } \binom{m}{j}   (m-j) \big(a \dashv_{(j)}  b \big) ~ \! t^{m+n-j-1 } +\sum_{j \geq 0 } \binom{m}{j}   (m-j) \big(a \dashv_{(j)}  b \big) ~ \! t^{m+n-j-1 }\\
        =&~ \alpha,
    \end{align*}
    where 
    \begin{align*}
    \alpha = \sum \limits_{j \geq 0 } \binom{m}{j}   \p \big(a \dashv_{(j)}  b \big) ~ \! t^{m+n-j } +\sum \limits_{j \geq 0 } \binom{m}{j}   (m+n-j) \big(a \dashv_{(j)}  b \big) ~ \! t^{m+n-j -1} 
    \end{align*}
    is an element of $(\partial+\partial_t)\D[t^{\pm1}]$. Therefore, $at^m \dashv \big( (\p b)t^n + nbt^{n-1} \big) \in (\partial+\partial_t)\D[t^{\pm1}]$. In a similar way, one may show that $\big(  (\p a)t^{m} + mat^{m-1} \big)\vdash bt^{n} $ and $at^m \vdash \big( (\p b)t^n + nbt^{n-1} \big) \in (\partial+\partial_t)\D[t^{\pm1}]$. This completes the proof.
\end{proof}

Let $\mathrm{Diass} (\D) = \D[t^{\pm1}]\big/(\partial+\partial_t) ~ \! \D[t^{\pm1}]$ be the quotient of the diassociative algebra $ \D[t^{\pm1}]$ by the two-sided ideal $(\partial+\partial_t) ~ \! \D[t^{\pm1}]$.
Then the operations $\dashv$ and $\vdash$ on $\D[t^{\pm1}]$ induce corresponding
operations on $\mathrm{Diass} (\D)$, making it a diassociative algebra. For any $a\in \D$ and $m\in \Z$, we denote by $a_m$ the image of $at^m$ in $\mathrm{Diass} (\D)$. 

\begin{remark}\label{anni-di}
Note that there is a diassociative algebra derivation $T: \mathrm{Diass} (\D) \rightarrow \mathrm{Diass} (\D)$ induced by $-\p_t$. Explicitly, we have $T(a_m)= -ma_{m-1}$. It is clear from (\ref{equation 11}) that the subspace spanned by all $a_m$ with $m \geq 0$ is a $T$-invariant subalgebra of the diassociative algebra $\mathrm{Diass} (\D)$. This is called the {\bf annihilation diassociative algebra} of the given diassociative conformal algebra $\D$, and we denote this by $\D_{-}$. 
\end{remark}

 Our aim now is to construct a family $\mathcal{F}$ of $\mathrm{Diass} (\D)$-valued formal distributions such that $(\mathrm{Diass} (\D), \mathcal{F})$ is a formal distribution diassociative algebra. For this, we set
\begin{align*}
\mathcal{F}
=
\big\{
a(w)=\sum_{n\in\mathbb Z}(a_{n})w^{-n-1} ~ \! \big| ~ \!
a\in \D
\big\}.
\end{align*}
Then we have the following result.

\begin{proposition}\label{p-j}
Let $\D$ be a diassociative conformal algebra.
Then the family $\mathcal{F}$ defined above is closed under all $j$-th products defined in (\ref{j-th}) and is invariant under the formal differentiation operator $\partial_w$. 
\end{proposition}

\begin{proof}
For any $a(w)=\sum \limits_{m\in\mathbb Z}(a_{m})w^{-m-1} \text{ and } b(w)=\sum \limits_{n\in\mathbb Z}(b_{n})w^{-n-1} \in \mathcal{F}$, it follows from (\ref{equation 11}) that
\begin{align*}
    a(z) \dashv b(w) = \sum_{k \geq 0} (a \dashv_{(k)} b)(w)\frac{\p_{w}^{k}\delta(z,w)}{k!}~\text{ and }~ a(z) \vdash b(w) = \sum_{k \geq 0} (a \vdash_{(k)} b)(w)\frac{\p_{w}^{k}\delta(z,w)}{k!},
\end{align*}
where $\delta(z,w) = \sum \limits_{ n \in \Z}z^{-n-1} w^{n}$ is the delta distribution. It is well-known \cite{kac-1} for the delta distribution that 
\begin{align*}
     \operatorname{Res}_z ( (z-w)^j \p_{w}^{k} \delta(z,w) ) =
     \begin{cases}
         j!,  & \text{ for } j=k, \\
         0, &  \text{ for } j \neq k.
     \end{cases}
 \end{align*}
Hence, for any $j \geq 0$, we have 
\begin{align*}
    a(w) \dashv_{(j)} b(w) = \operatorname{Res}_z (z-w)^j ~ \! a(z) \dashv b(w)= (a \dashv_{(j)} b )(w),\\
    a(w) \vdash_{(j)} b(w) = \operatorname{Res}_z (z-w)^j ~ \! a(z) \vdash b(w)= (a \vdash_{(j)} b )(w).
\end{align*}
This shows that the family $\mathcal{F}$ is closed under all $j$-th products defined in (\ref{j-th}). To show that $\mathcal{F}$ is $\p_w$-invariant, let $a(w)=\sum \limits_{n\in\mathbb Z}(a_{n})w^{-n-1}$ be any element of $\mathcal{F}$. Then, using the equality $(\partial a)_n =- n a_{n-1}$ in $D$, we have  
\begin{align*}
    \p_w (a(w)) = \sum \limits_{n\in\mathbb Z} (-n-1)(a_{n})w^{-n-2}=\sum \limits_{n\in\mathbb Z} (-na_{n-1})w^{-n-1}=\sum \limits_{n\in\mathbb Z} ((\p a)_{n})w^{-n-1}=(\p a)(w).
\end{align*}
This completes the proof.
\end{proof}

By our construction, the coefficients of all distributions in $\mathcal{F}$ span $\mathrm{Diass} (\D)$. Next, let $a(w)=\sum \limits_{m\in\mathbb Z}(a_m)w^{-m-1}$ and
$b(w)=\sum \limits_{n\in\mathbb Z}(b_n)w^{-n-1}$ be arbitrary elements of $\mathcal{F}$. From the definitions of $\dashv$ and $\vdash$ on $\mathrm{Diass} (\D)$, and the corresponding $j$-th products of formal distributions, we get that
\begin{align*}
a(z)\dashv b(w)
=\sum_{j\geq 0}(a\dashv_{(j)}b)(w) ~ \! 
\frac{\partial_w^j\delta(z,w)}{j!}~\text{ and }~
a(z)\vdash b(w)
=\sum_{j\geq 0}(a\vdash_{(j)}b)(w) ~ \!
\frac{\partial_w^j\delta(z,w)}{j!}.
\end{align*}
Hence, there exists a non-negative integer $N$ such that $(z-w)^N a(z)\dashv b(w)=0=(z-w)^N a(z)\vdash b(w)$. Thus, $\big(a(w),b(w)\big)$ is a local pair. Consequently, $(\mathrm{Diass} (\D) ,\mathcal{F})$ is a formal distribution diassociative algebra. Furthermore, Proposition \ref{p-j} shows that the smallest subspace of $\mathrm{Diass} (\D) \llbracket w^{\pm 1 }\rrbracket$ containing $\mathcal{F}$ which is closed under all $j$-th products (\ref{j-th}) and $\p$-invariant coincides with $\mathcal{F}$ itself, that is, $\overline{\mathcal{F}}=\mathcal{F}.$ Hence, by the construction of the corresponding diassociative conformal algebra, we get that $\text{Conf} ( \mathrm{Diass} (\D) ,\mathcal{F})=\overline{\mathcal{F}}=\mathcal{F}. $ Further, the map $\varphi:\mathcal{D}\rightarrow\mathcal{F}$ given by $\varphi(a) :=a(w)=\sum \limits_{m\in\mathbb Z}(a_m)w^{-m-1}$ is an isomorphism of diassociative conformal algebras. Hence, for the given diassociative conformal algebra $\mathcal{D}$, we have $ \operatorname{Conf}(\mathrm{Diass} (\D) , \mathcal{F}) \cong \D. $  

\medskip
    
    Two formal distribution diassociative algebras $(D_1,\mathcal{F}_1)$ and $(D_2,\mathcal{F}_2)$ are said to be {\bf equivalent} if their corresponding diassociative conformal algebras $\operatorname{Conf}(D_1,\mathcal{F}_1)$ and $\operatorname{Conf}(D_2,\mathcal{F}_2)$ are isomorphic. Hence, we conclude the following.

\begin{theorem}\label{thm-categ-equiv}
    The category of diassociative conformal algebras and the category of equivalence classes of formal distribution diassociative algebras are equivalent.
\end{theorem}

\section{Maurer-Cartan characterization and cohomology of diassociative conformal algebras}\label{sec3} In this section, we first construct a graded Lie algebra whose Maurer-Cartan elements are in one-to-one correspondence with diassociative conformal algebra structures on a given $\C [\partial]$-module. Subsequently, we define the cohomology of a diassociative conformal algebra. We show that the graded space of cohomology inherits a natural Gerstenhaber algebra structure. Later, we extend the construction to define the cohomology with coefficients in an arbitrary representation.

\medskip

To achieve the Maurer-Cartan characterization of diassociative conformal algebras, we first recall some basic notions concerning planar binary trees and nonsymmetric operads. A planar binary tree with $n$ internal vertices, often called an {\em $n$-tree}, is a rooted planar tree consisting of $n+1$ leaves and a single root, where every internal vertex has exactly two incoming edges and one outgoing edge. The planar structure specifies a fixed left-to-right ordering of the branches at each vertex. For each $n \geq 0$, let $Y_{n}$ be the set of all planar binary trees with $n$ internal vertices. In the degenerate case $n=0$, the set $Y_{0}$ consists of the unique tree containing only a single root. For the first few values of $n$, the corresponding sets $Y_{n}$ are illustrated below:

\medskip

\[
\begin{aligned}
Y_0 &= \Bigg\{~ \vcenter{\hbox{\begin{tikzpicture}[scale=0.2, baseline={(0,0)}]    
\draw (0,0) -- (0,4);
\end{tikzpicture}}}~ \Bigg\}, \quad
Y_1 = \Bigg\{ \vcenter{\hbox{\begin{tikzpicture}[scale=0.2, baseline={(0,-0.5)}]
\draw (0,0)-- (0,-2); \draw (0,0) -- (-2,2); \draw (0,0) -- (2,2);
\end{tikzpicture}}} \Bigg\}, \quad
Y_2 = \Bigg\{ \vcenter{\hbox{\begin{tikzpicture}[scale=0.2, baseline={(0,-0.5)}]
\draw (0,0) -- (2,-2); \draw (2,-2) -- (4,0); \draw (2,-2) -- (2,-4); \draw (1,-1) -- (2,0); 
\end{tikzpicture}}} ,~ 
\vcenter{\hbox{\begin{tikzpicture}[scale=0.2, baseline={(0,-0.5)}]
\draw (0,0) -- (2,-2);    \draw (2,-2) -- (4,0);     \draw (2,-2) -- (2,-4);     \draw (3,-1) -- (2,0);
\end{tikzpicture}}}
\Bigg\}, \\[1em]
Y_3 &= \Bigg\{
\vcenter{\hbox{\begin{tikzpicture}[scale=0.2, baseline={(0,-0.5)}]
\draw (0,0)-- (2,-2); \draw (2,-2) -- (2,-4); \draw (2, -2) -- (4,0); 
\draw (0.7, - 0.7) -- (1.3333, 0) ; \draw (1.33, -1.33) -- (2.66, 0);
\end{tikzpicture}}} ,~
\vcenter{\hbox{\begin{tikzpicture}[scale=0.2, baseline={(0,-0.5)}]
\draw (0,0) -- (2,-2); \draw (2, -2) -- (2, -4); \draw (2, -2) -- (4,0); 
\draw (1.33, -1.33) -- (2.66, 0); \draw (1.33, 0) -- (2, -0.66);
\end{tikzpicture}}} ,~
\vcenter{\hbox{\begin{tikzpicture}[scale=0.2, baseline={(0,-0.5)}] 
\draw (0,0)-- (2,-2); \draw (2,-2)-- (2,-4); \draw (2,-2) -- (4,0); 
\draw (1.33, 0) -- (0.66, -0.66); \draw (2.66, 0) -- (3.34, -0.66); 
\end{tikzpicture}}} ,~
\vcenter{\hbox{\begin{tikzpicture}[scale=0.2, baseline={(0,-0.5)}]  
\draw (0,0) -- (2,-2); \draw (2, -2) -- (2, -4); \draw (2,-2) -- (4, 0); 
\draw (1.33, 0) -- (2.67 , -1.33) ; \draw (2.66, 0) -- (2, -0.66); 
\end{tikzpicture}}} ,~
\vcenter{\hbox{\begin{tikzpicture}[scale=0.2, baseline={(0,-0.5)}]  
\draw (0,0) -- (2, -2) ; \draw (2, -2) -- (2, -4) ; \draw (2, -2) -- (4, 0); 
\draw (1.33, 0) -- (2.67, - 1.33) ; \draw (2.66, 0) -- (3.34, -0.66);	
\end{tikzpicture}}}
\Bigg\}, \quad \text{etc.}
\end{aligned}
\]

\noindent Note that the cardinality of $Y_{n}$ is given by the $n$-th Catalan number $\frac{(2n)!}{(n+1)!~ n!}$. There are some fundamental operations on the set of all planar binary trees that play a central role in our construction (see \cite{Frabetti,loday1,maj-muk}).
\begin{itemize}
    \item[(i)] {\bf Face maps.} Let $y \in Y_n$ be a planar binary tree. The $n+1$ leaves of $y$ are labelled consecutively from left to right as $0, 1, \ldots, n$. For every $n \geq 1$ and $0 \leq i \leq n$, the {\em $i$-th face map} is the function $d_{i}: Y_{n} \to Y_{n-1},~ y \mapsto d_{i}y$ obtained by deleting the $i$-th leaf of $y$. These face maps satisfy the simplicial identity $d_{i} \circ d_{j} = d_{j-1} \circ d_{i}, \mathrm{ for~} i < j$.
    
    \item[(ii)] {\bf Grafting operation.} Let $y_{1}\in Y_{m}$ and $y_{2}\in Y_{n}$ be two planar binary trees. Their grafting, denoted by $y_{1} \vee y_{2}$, is defined to be the planar binary tree in $Y_{m+n+1}$ obtained by attaching the roots of $y_{1}$ and $y_{2}$ to a new internal vertex and declaring the edge emanating from this vertex as the new root. This operation provides a natural recursive procedure for constructing larger planar binary trees from smaller ones.
    
    \item[(iii)] {\bf Orientation map.} For every $n \geq 1$ and $0 \leq i \leq n$, there is a map $\bullet_{i}:~Y_{n}\to  \{\dashv, ~ \! \vdash\},~ \! y \mapsto \bullet_{i}^{y}$, defined by
    \begin{align*}
\bullet_0^y
&=
\begin{cases}
\dashv, & \text{if } y = | \vee y_1 \text{ for some } (n-1)\text{-tree } y_1,\\
\vdash, & \text{otherwise},
\end{cases}
\\[2ex]
\bullet_i^y
&=
\begin{cases}
\dashv, & \text{if the } i\text{-th leaf of } y \text{ is oriented like ``}\backslash\text{''},\\
\vdash, & \text{if the } i\text{-th leaf of } y \text{ is oriented like ``/''},
\end{cases}
\qquad (1\le i\le n-1)
\\[2ex]
\bullet_n^y
&=
\begin{cases}
\vdash, & \text{if } y = y_1 \vee | \text{ for some } (n-1)\text{-tree } y_1,\\
\dashv, & \text{otherwise}.
\end{cases}
\end{align*}
    
\item[(iv)] {\bf Structure maps.} Let $m,n \geq 1 \mathrm{~and~}1\leq i \leq m $. Two auxiliary maps $R_{0}^{m; i,n}: Y_{m+n-1} \rightarrow Y_{m}$ and $R_{i}^{m;i,n}: Y_{m+n-1} \to Y_{n}$ are defined as compositions of the face maps as follows:
\begin{align*}
R_0^{m; i, n}  = \widehat{d_0} \circ \widehat{d_1} \circ \cdots \circ \widehat{d_{i-1}} \circ d_i \circ \cdots \circ d_{i+n-2} \circ \widehat{d_{i + n-1}} \circ \cdots \circ \widehat{d_{m+n-1}},\\
R_i^{m; i, n}  = d_0 \circ d_1 \circ \cdots \circ d_{i-2} \circ \widehat{d_{i-1}} \circ \cdots \circ \widehat{d_{i+n-1}} \circ d_{i+n} \circ \cdots \circ d_{m+n-1},
\end{align*}
where the notation $\widehat{d_j}$ indicates that the corresponding face map $d_j$ is omitted from the expression.
\end{itemize}

\medskip

Let $\D$ and $\mathcal{D'}$ be two $\C[\p]$-modules, not necessarily having the structure of diassociative conformal algebras. For any $n \geq 1$, a $\C$-linear map 
\begin{align*}
f : \C[Y_n]\otimes \mathcal{ D}^{\otimes n} \rightarrow \mathcal{ D'}[\la_1,\ldots , \la _{n-1}],~ ~y\otimes a_1 \otimes \cdots \otimes a_n \mapsto f_{\la_1 , \ldots , \la _{n-1}}(y;a_1 , \ldots , a _n)
\end{align*}
is called a {\em conformal sesquilinear map} if it satisfies 
\begin{align*}
     f_{\la_1 , \ldots , \la _{n-1}}(y;a_1 , \ldots , \p a_i , \ldots, a _n) =~
     \begin{cases}
        -\la_i f_{\la_1 , \ldots , \la _{n-1}}(y;a_1 , \ldots ,a_{n-1}, a _n), \quad \text{if } 1\leq i \leq n-1, \\
         (\p + \la_1 + \cdots + \la_{n-1}) ~ \! f_{\la_1 , \ldots , \la _{n-1}}(y;a_1 , \ldots , a_{n-1}, a _n), \quad \text{if } i=n,
     \end{cases}
 \end{align*}
   for all $y \in Y_n$ and $a_1, \ldots , a_n \in \D$. We denote by $ C_{cY}^{n}(\D,\mathcal{D'})$ the $\C$-vector space consisting of all conformal sesquilinear maps from $\C[Y_n]\otimes \D^{\otimes n} \rm{~to~} \mathcal{D'}[\la _1 , \ldots , \la_{n-1}]$. In the special case $n=1$, the set $Y_{1}$ consists of a single planar binary tree, and there are no formal variables $\la_{i}$'s. Consequently, the conformal sesquilinearity condition reduces to the requirement that the map commutes with the action of $\p$. Hence, we have
$C_{cY}^{1}(\D,\D')
=\operatorname{Hom}_{\C[\p]}(\D,\D').$

\medskip

From now on, let $\D$ be a fixed $\C[\p]$-module. For every pair of positive integers $m,n$ and each $1 \leq i \leq m$, we define a map (called {\em partial composition}) 
        \begin{align*}
        \circ_i : C_{cY}^{m}(\D,\D) \otimes C_{cY}^{n}(\D,\D) \longrightarrow C_{cY}^{m+n-1}(\D,\D)
        \end{align*}
        by 
    \begin{align}
    \label{equ A}   
    &(f\circ _i g)_{\la_1 , \ldots , \la_{m+n-2}}(y;a_1,\ldots , a_{m+n-1}) \nonumber \\
    =& f_{\la _1, \ldots , \la_{i-1},\la_{i}+\cdots+\la_{i+n-1},\la_{i+n}, \ldots , \la_{m+n-2}} \big( R_{0}^{m;i,n}y;a_1, \ldots , g_{\la_i,\ldots,\la_{i+n-2}}(R_{i}^{m;i,n}y;a_i, \ldots, a_{i+n-1}),\ldots, a_{m+n-1} \big), 
 \end{align}
for $f \in C_{cY}^m(\D,\D),~ g \in C_{cY}^n(\D,\D), ~ y \in Y_{m+n-1} \text{ and } a_1, \ldots, a_{m+n-1} \in \D$.  In the following result, we aim to show that the collection of $\C$-vector spaces $\{C_{cY}^{n}(\D, \D)\}_{n \geq 1}$ equipped with the partial compositions (\ref{equ A}) forms a nonsymmetric operad. We first recall the following.

\medskip

\begin{definition}
    A {\bf nonsymmetric operad} is a triple $\mathcal{O} = (\{\mathcal{O}(n)\}_{n \geq 1}, \circ,  \mathbbm{1} )$ consisting of a collection $\{\mathcal{O}(n)\}_{n \geq 1} $ of $\C$-vector spaces equipped with $\C$-linear maps (called {\em partial compositions})
    \begin{align*}
        \circ_i :\mathcal{O}(m)\otimes \mathcal{O}(n) \longrightarrow \mathcal{O}(m+n-1),  \text{ for }  1 \leq i \leq m 
    \end{align*}
   such that for any $f \in \mathcal{O}(m), ~ \! g \in \mathcal{O}(n) \mathrm{~and~}  h \in \mathcal{O}(p)$, 
    \begin{align}
        (f \circ_i g) \circ_{i+j-1} h =~& f \circ_i (g \circ_j h),\text{ for } 1 \leq i \leq m ,~ \!  1 \leq j \leq n, \label{a-1}\\
        (f \circ_i g) \circ_{j+n-1} h =~& (f \circ_j h ) \circ_i g,\text{ for } 1 \leq i < j \leq m \label{a-2}
\end{align}
and there is an element $\mathbbm{1}\in \mathcal{O}(1)$ such that $f \circ_i \mathbbm{1} = \mathbbm{1} \circ_1 f= f$, for all $f \in \mathcal{O}(m) $ and $1 \leq i \leq m$.
\end{definition} 

 \begin{thm}\label{thm-operad}
     Let $\D$ be a $\C[\p]$-module. Then the collection $\{C_{cY}^{n}(\D,\D)\}_{n \geq 1}$ of $\C$-vector spaces together with the partial compositions defined in (\ref{equ A}), forms a nonsymmetric operad. The distinguished element $\mathbbm{1} \in C_{cY}^{1}(\D,\D)$ is given by 
     $\mathbbm{1}( \begin{tikzpicture}[scale=.3]
\draw (0,0)-- (0,-.5); \draw (0,0) -- (-.5,.5); \draw (0,0) -- (.5,.5);
\end{tikzpicture};a) =a$, for all $a \in \D$. 
 \end{thm}
        
\begin{proof}
Let $f \in C_{cY}^{m}(\D,\D), \, g \in C_{cY}^{n}(\D,\D), \, h \in C_{cY}^{p}(\D,\D)$  with $1 \leq i \leq m $ and $1 \leq j \leq n$. Then for any $y \in Y_{m+n+p-2}$ and $a_1, \ldots, a_{m+n+p-2} \in \D$, a direct computation shows that

\begin{align} \label{eqi-1}
  & \big((f \circ_{i}g)\circ_{i+j-1} h \big)_{\la_{1}, \ldots, \la_{m+n+p-3}}(y; a_{1}, \ldots ,a_{m+n+p-2}) \nonumber \\
=~ & (f \circ _{i} g)_{\la_{1}, \ldots ,\la_{i+j-2},\la_{i+j-1}+ \cdots + \la_{i+j+p-2},\la_{i+j+p-1}, \ldots , \la_{m+n+p-3} } \big(R_{0}^{m+n-1;i+j-1,p}y;a_{1}, \ldots a_{i+j-2}, \nonumber \\
  & \qquad \qquad  h_{\la_{i+j-1}, \ldots \la_{p+i+j-3}}(R_{i+j-1}^{m+n-1;i+j-1,p}y;a_{i+j-1}, \ldots, a_{P+i+j-2}),a_{p+i+j-1},\ldots, a_{m+n+p-2} \big) \nonumber \\
 =~ & f_{\la_{1}, \ldots \la_{i-1},\la_{i} + \cdots + \la_{n+p+i-2},\la_{n+p+i-1}, \ldots ,\la_{m+n+p-3} }\big(R_{0}^{m;i;n}R_{0}^{m+n-1;i+j-1,p}y;a_{1}, \ldots a_{i-1}, \nonumber \\
  &  \qquad \qquad  g_{\la_{i}, \ldots , \la_{i+j-2}, \la_{i+j-1}+ \cdots +\la_{p+i+j-2}, \la_{p+i+j-1},\ldots, \la_{n+p+i-3}}\big(R_{i}^{m;i,n}R_{0}^{m+n-1;i+j-1,p}y; a_{i},\ldots a_{i+j-2}, \nonumber \\
  &  \qquad \qquad \quad h_{\la_{i+j-1}, \ldots \la_{p+i+j-3}}(R_{i+j-1}^{m+n-1;i+j-1,p}y;a_{i+j-1}, \ldots, a_{P+i+j-2}),a_{p+i+j-1}, \ldots , a_{n+p+i-2} \big), \nonumber \\
  &  \qquad \qquad \quad \quad  a_{n+p+i-1}, \ldots  , a_{m+n+p-2} \big)
\end{align}
  and
  \begin{align*} 
    & \big( f \circ _{i}(g \circ _{j}h) \big)_{\la_{1}, \ldots , \la_{m+n+p-3}} (y;a_{1}, \ldots , a_{m+m+p-2}) \nonumber \\
    =~ & f_{\la_{1}, \ldots , \la_{i-1},\la_{i}+\cdots + \la_{n+p+i-2},\la_{n+p+i-1},\ldots , \la_{m+n+p-3}} \big( R_{0}^{m;i,n+p-1}y;a_{1}, \ldots a_{i-1}, \nonumber \\
     & \qquad \big(g \circ _{j} h \big)_{\la_{i}, \ldots , \la _{n+p+i-3}}(R_{i}^{m;i,n+p-1}y;a_{i}, \ldots , a_{n+p+i-2}), a_{n+p+i-1}, \ldots , a_{m+n+p-2} \big) \nonumber 
     \end{align*}
     \begin{align}\label{eqi-2}
    =~ & f_{\la_{1}, \ldots , \la_{i-1},\la_{i}+\cdots + \la_{n+p+i-2},\la_{n+p+i-1},\ldots , \la_{m+n+p-3}} \big( R_{0}^{m;i,n+p-1}y;a_{1}, \ldots a_{i-1}, \nonumber \\
    & \qquad g_{\la_{i}, \ldots , \la_{i+j-2}, \la_{i+j-1}+ \cdots +\la_{p+i+j-2}, \la_{p+i+j-1},\ldots, \la_{n+p+i-3}} \big( R_{0}^{n;j,p}R_{i}^{m;i,n+p-1}y; a_{i},\ldots a_{i+j-2}, \nonumber \\
    & \qquad \quad h_{\la_{i+j-1}, \ldots , \la _{p+i+j-3}}(R_{j}^{n;j,p}R_{i}^{m;i,n+p-1}y; a_{i+j-1}, \ldots, a_{p+i+j-2}),a_{p+i+j-1}, \ldots , a_{n+p+i-2} \big), \nonumber \\
    & \qquad \qquad a_{n+p+i-1}, \ldots  , a_{m+n+p-2} \big).
  \end{align} 
  We note that, for any $y \in Y_{m+n+p-2}$, the structure maps $R_0^{m;i, n}$'s and $R_i^{m;i, n}$'s satisfy
  \begin{align*}
 &R_{0}^{m;i;n}R_{0}^{m+n-1;i+j-1,p}y= R_{0}^{m;i,n+p-1}y, \quad
  R_{i}^{m;i,n}  R_{0}^{m+n-1;i+j-1,p}y = R_{0}^{n;j,p} R_{i}^{m;i,n+p-1}y \\ & \qquad \qquad \qquad \text{ and } ~~ R_{i+j-1}^{m+n-1;i+j-1,p}y= R_{j}^{n;j,p}R_{i}^{m;i,n+p-1}y.
\end{align*}
Substituting these identities into (\ref{eqi-1}) and (\ref{eqi-2}) shows that $(f \circ_i g) \circ_{i+j-1} h = f \circ_i (g \circ_j h)$. Similarly, for $1 \leq i < j \leq m$, one can check that $(f \circ_i g) \circ_{j+n-1} h = (f \circ_j h ) \circ_i g$.

Finally, we verify the unit axiom. Using the identities $R_{0}^{m;i,1}=\mathrm{id}=R_{1}^{1;1,m}
 \text{ and } 
R_{0}^{1;1,m}=
\begin{tikzpicture}[scale=.3]
\draw (0,0)-- (0,-.5);
\draw (0,0) -- (-.5,.5);
\draw (0,0) -- (.5,.5);
\end{tikzpicture},$ it follows that, for every $f\in C_{cY}^{m}(\D,\D)$ and $1\leq i\leq m$,
\begin{align*}
(f\circ_i\mathbbm{1})_{\la_1,\ldots,\la_{m-1}}
(y;a_1,\ldots,a_m)
&=
f_{\la_1,\ldots,\la_{m-1}}
(
R_0^{m;i,1}y;
a_1,\ldots,a_{i-1},
\mathbbm{1}(R_i^{m;i,1}y;a_i),
a_{i+1},\ldots,a_m
)\\
&=
f_{\la_1,\ldots,\la_{m-1}}
(y;a_1,\ldots,a_m),
\end{align*}
\begin{align*}
(\mathbbm{1}\circ_1f)_{\la_1,\ldots,\la_{m-1}}
(y;a_1,\ldots,a_m)
&=
\mathbbm{1}
(
R_0^{1;1,m}y;
f_{\la_1,\ldots,\la_{m-1}}
(R_1^{1;1,m}y;a_1,\ldots,a_m)
)\\
&=
f_{\la_1,\ldots,\la_{m-1}}
(y;a_1,\ldots,a_m).
\end{align*}
This shows that $f\circ_i\mathbbm{1}
=
\mathbbm{1}\circ_1f
=
f$. This proves the theorem.
\end{proof}

Using the partial compositions $\circ_i$ defined in (\ref{equ A}), one may define a graded bracket of degree $-1$ 
\begin{align*}
[ ~,~]:C_{cY}^{m}(\D,\D) \otimes C_{cY}^{n}(\D,\D) \longrightarrow C_{cY}^{m+n-1}(\D,\D), \text{ for } m,n \geq 1, \text{ by }
\end{align*}
\begin{align} \label{brac-form}
    [f,g] := \sum_{i =1}^{m} (-1)^{(i-1)(n-1)} ~ \! f \circ_i g - (-1)^{(m-1)(n-1)} \sum_{i=1}^{n} (-1)^{(i-1)(m-1)} ~ \! g \circ_i f.
\end{align}
for $f \in C_{cY}^{m}(\D,\D)$ and $g \in C_{cY}^{n}(\D,\D)$.
It turns out that \cite{gers-ring,gers-voro} the above bracket satisfies the following shifted graded Jacobi identity:

\begin{align*}
    [f,[g,h]] = [[f,g],h]+ (-1)^{(m-1)(n-1)}  ~ \! [g,[f,h]],
\end{align*}
for all $f \in C_{cY}^{m}(\D,\D), g \in C_{cY}^{n}(\D,\D) \text{ and } h \in C_{cY}^{p}(\D,\D)$. Consequently, the shifted graded vector space $\bigoplus\limits_{n \geq 0}C_{cY}^{n+1}(\D,\D)$ endowed with the bracket $[~,~]$ is a graded Lie algebra. The significance of this graded Lie algebra is established in the following proposition.

\begin{proposition} \label{1-1}
    Let $\D$ be a $\C[\p]$-module. Then diassociative conformal algebra structures on $\D$ are in {one-to-one} correspondence with Maurer-Cartan elements of the graded Lie algebra $\big( \bigoplus \limits_{n \geq 0}C_{cY}^{n+1}(\D,\D),~[~,~] \big)$.
\end{proposition}

\begin{proof}
We observe that an element $\pi \in C_{cY}^{2}(\D,\D)$ is equivalent to have two conformal sesquilinear $\la$-products $\dashv_\la , ~ \! \vdash_\la : \D \times \D \rightarrow \D [\la]$ given by
\begin{align}\label{pi-lambda}
    a \dashv _{\la} b = \pi_{\la} \big(\tikz[baseline=-0.5ex,scale=0.20]{
  \draw (-1,1)--(0,0);
  \draw (1,1)--(0,0);
  \draw (0,0)--(0,-1);
  \draw (0.5,0.5)--(0,1);
};a,b \big) ~ \text{ and } ~  a \vdash _{\la} b=\pi_{\la} \big(\tikz[baseline=-0.5ex,scale=0.20]{
  \draw (-1,1)--(0,0);
  \draw (1,1)--(0,0);
  \draw (0,0)--(0,-1);
  \draw (-0.5,0.5)--(0,1);
};  a,b \big), \text{ for }a,b \in \D.
\end{align}
Then for any $y \in Y_3$ and $a, b, c \in \D$, a straightforward computation shows that
\begin{align} \label{id}
     \frac{1}{2}[\pi, \pi]_{\la, \mu}(y; a,b,c ) =~
     \begin{cases}
         (a \vdash_{\la}b) \vdash_{\la + \mu}c - a \vdash_{\la}(b \vdash _{\mu}c), \quad \text{ if } y = \begin{tikzpicture}[scale=0.1]
	     \draw (12,0)-- (14,-2); \draw (14,-2) -- (14,-4); \draw (14, -2) -- (16,0); \draw (12.7, - 0.7) -- (13.3333, 0) ; \draw (13.33, -1.33) -- (14.66, 0);
	     \end{tikzpicture},\\
         (a \dashv_{\la} b) \vdash_{\la + \mu }c - a \vdash_{\la}(b \vdash _{\mu } c),  \quad \text{ if } y = \begin{tikzpicture}[scale=0.1]
 \draw (18,0) -- (20,-2); \draw (20, -2) -- (20, -4); \draw (20, -2) -- (22,0); \draw (19.33, -1.33) -- (20.66, 0); \draw (19.33, 0) -- (20, -0.66);
\end{tikzpicture},\\
(a \vdash _{\la} b) \dashv_{\la + \mu } c - a \vdash_{\la} ( b \dashv_{\mu}c), \quad \text{ if } y = \begin{tikzpicture}[scale=0.1] 
  \draw (24,0)-- (26,-2); \draw (26,-2)-- (26,-4); \draw (26,-2) -- (28,0); \draw (25.33, 0) -- (24.66, -0.66); \draw (26.66, 0) -- (27.34, -0.66); ,
\end{tikzpicture},\\
(a \dashv _{\la} b) \dashv_{\la + \mu} c - a \dashv_\la (b \vdash _{\mu} c), \quad \text{ if } y = \begin{tikzpicture}[scale=0.1]  
  \draw (30,0) -- (32,-2); \draw (32, -2) -- (32, -4); \draw (32,-2) -- (34, 0); \draw (31.33, 0) -- (32.67 , -1.33) ; \draw (32.66, 0) -- (32, -0.66); 
\end{tikzpicture},\\
( a \dashv _{\la}b) \dashv_{\la + \mu}c - a \dashv_{\la} (b \dashv _{\mu}c),\quad \text{ if }y= \begin{tikzpicture}[scale=0.1]  
   \draw (36,0) -- (38, -2) ; \draw (38, -2) -- (38, -4) ; \draw (38, -2) -- (40, 0); \draw (37.33, 0) -- ( 38.67, - 1.33) ; \draw (38.66, 0) -- (39.34, -0.66);	
\end{tikzpicture}.
     \end{cases}
\end{align}
This shows that $\pi$ is a Maurer-Cartan element of the graded Lie algebra $\big( \bigoplus \limits_{n \geq 0}C_{cY}^{n+1}(\D,\D),~[~,~] \big)$ if and only if $(\D, \dashv_\la, \vdash_\la)$ is a diassociative conformal algebra.
\end{proof}

\medskip

\medskip

\medskip

\noindent {\bf Cohomology of a diassociative conformal algebra.}

Let $(\D, \dashv_{\la}, \vdash_{\la})$ be a diassociative conformal algebra. For each integer $n \geq 0 $, we define the space of $n$-cochains by
 \begin{align*}
     C_{cDiass}^{n}(\D,\D) :=
     \begin{cases}
         \D/{\p \D},  & \text{ for } n=0, \\
         C_{cY}^{n}(\D,\D), &  \text{ for } n \geq 1.
     \end{cases}
 \end{align*}
Then there is a map $ \delta_{cDiass} : C_{cDiass}^{n}(\D,\D) \rightarrow C_{cDiass}^{n+1}(\D,\D)$ explicitly given by
\begin{align*}
\delta_{cDiass}(a + \p \D)( b) := (b \dashv_{-\p -\la} a - a \vdash _{\la} b) \big|_{ \la  = 0}, \text{ for } a + \p \D \in \D / \partial \D = C^0_{cDiass} (\D, \D) \text{ and } b \in \D,
\end{align*}
   \begin{align*}
{\delta_{cDiass}(f)}_{\la_{1}, \ldots, \la_{n}}&(y; a_1, \ldots , a_{n+1}) :=~ 
a_{1} (\bullet_{0}^{y})_{ \la_{1}}f_{ \la_{2}, \ldots, \la_{n}}(d_{0}y;  a_{2}, \ldots , a_{n+1}) \\& ~+ \sum_{i =1}^{n}(-1)^{i} ~ \! f_{ \la_{1}, \ldots , \la_{i-1}, \la_{i}+ \la_{i+1}, \la_{i+2}, \ldots , \la_{n}}(d_{i}y;  a_{1}, \ldots, a_{i-1},a_{i}(\bullet_{i}^{y})_{\la_{i}} a_{i+1}, \ldots, a_{n+1}) \\&
~+(-1)^{n+1} ~ \!  f_{\la_{1}, \ldots, \la_{n-1}}(d_{n+1}y; a_{1}, \ldots, a_{n})(\bullet_{n+1}^{y})_{\la_{1}+ \cdots+ \la_{n} }a_{n+1},
\end{align*}
for $f \in C_{cDiass}^{n \geq 1}(\D,\D)$, $ y \in Y_{n+1} \text{ and } a_1 , \ldots , a_{n+1} \in \D$. For any $ a + \p \D \in C^0_{cDiass} (\D, \D)$ and $b , c \in \D$, we observe that
\begin{align*}
    &{\delta_{cDiass}^{2}(a+\p \D)}_{\la} \left( \tikz[baseline=-0.5ex,scale=0.20]{
  \draw (-1,1)--(0,0);
  \draw (1,1)--(0,0);
  \draw (0,0)--(0,-1);
  \draw (0.5,0.5)--(0,1);
}; b, c \right) \\
=~& b \dashv_{\la} (c \dashv_{-\p -\mu}a - a \vdash_{ \mu}c) \big|_{\mu = 0} - \big( ( b \dashv _{\la} c) \dashv_{ -\p -\mu} a - a \vdash_{\mu}(b \dashv_{\la} c) \big) \big|_{ \mu =0} \\ 
& \qquad \qquad + (b \dashv_{-\p -\mu} a - a \vdash _{\mu}b ) \big|_{\mu = 0} \dashv_{\la} c\\
=~& b \dashv_{\la} (c \dashv_{-\p}a)  - (b \dashv_{\la}c)\dashv_{-\p}a + (b \dashv_{-\p }a)\dashv_{\la}c - b \dashv_{\la}(a \vdash_{(0)}c) \\
& \qquad \qquad +a \vdash_{(0)}(b \dashv_{\la}c) - (a\vdash_{(0)}b)  \dashv_{\la}c\\=~&
0 \quad (\text{by } (\ref{edi-1}), (\ref{edi-2}) \text{ and } (\ref{edi-3})).
\end{align*}
A similar computation shows that $ \delta_{cDiass}^{2}(a+\p \D)_{\la} \big( \tikz[baseline=-0.5ex,scale=0.20]{
  \draw (-1,1)--(0,0);
  \draw (1,1)--(0,0);
  \draw (0,0)--(0,-1);
  \draw (-0.5,0.5)--(0,1);
}; b, c \big)   = 0$. Hence we have ${\delta_{cDiass}^{2}(a+\p \D)} = 0$. For cochains of positive degree, the map $\delta_{cDiass}$ admits a description in terms of the graded Lie bracket introduced earlier. Indeed, for $f\in C_{cDiass}^{n \geq 1 }(\D,\D),$ we have
 \begin{align}\label{diff-formula}
 \delta_{cDiass}(f) = (-1)^{n-1} ~ \! [\pi,f],
 \end{align}
where $\pi\in C_{cY}^{2}(\D,\D)$ denotes the Maurer-Cartan element corresponding to the diassociative conformal algebra structure $(\D, \dashv_\la, \vdash_\la)$. Since $\pi$ satisfies the Maurer-Cartan equation $[\pi,\pi]=0$, the graded Jacobi identity immediately implies that $\delta_{cDiass}^{2} (f)= 0$. Consequently, $\{C_{cDiass}^{\bullet}(\D,\D), \delta_{cDiass}\}$ forms a cochain complex. The corresponding cohomology groups are called the {\bf cohomology groups} of the given diassociative conformal algebra $\D$, and the cohomology groups are denoted by $H_{cDiass}^{\bullet}(\D, \D)$. From the graded Jacobi identity of the bracket $[~, ~]$, and the formula of the differential (\ref{diff-formula}), it follows that the bracket $[~, ~]$ induces a degree $-1$ bracket (denoted by the same notation) on the graded cohomology space $H_{cDiass}^{\bullet}(\D, \D)$.

\medskip

Next, we show that the graded space of cohomology groups of a diassociative conformal algebra inherits a Gerstenhaber algebra structure. We first recall the following definition \cite{gers-voro,Hou-Shen-Zhao}.

\begin{definition}
A \textbf{Gerstenhaber algebra} is a triple $(\GG, \smile , [~, ~ ])$ consisting of a graded $\C$-module $\GG=\bigoplus \limits_{i\in\mathbb{Z}}\GG^{i}$ equipped with a degree $0$ graded-commutative associative product $\smile ~ \! :\GG^{i}\otimes\GG^{j}\rightarrow\GG^{i+j}$ and a degree $-1$ bracket $[\, ,\, ]:\GG^{i}\otimes\GG^{j}\rightarrow\GG^{i+j-1}$ subject to satisfy the following conditions:
\begin{align}
[x,y] &= -(-1)^{(i-1)(j-1)}[y,x],\\
[x,[y,z]] &= [[x,y],z]+(-1)^{(i-1)(j-1)}[y,[x,z]],\\
[x,y\smile z] &= [x,y]\smile z+(-1)^{(i-1)j} ~ \! y\smile[x,z], \label{leibn-rule}
\end{align}
for all homogeneous elements $x\in\GG^{i}$, $y\in\GG^{j}$ and $z\in\GG^{k}$.
\end{definition}

Let $(\D, \dashv_\la, \vdash_\la)$ be a diassociative conformal algebra with the corresponding Maurer-Cartan element $\pi$ (explicitly given in (\ref{pi-lambda})). We observe that the Maurer-Cartan equation $[\pi , \pi ] = 0$ is equivalent to the condition $\pi \circ_1 \pi = \pi \circ_2 \pi$. In other words, $\pi \in C^2_{cY} (\D, \D)$ is a {\em multiplication} on the nonsymmetric operad $(\{C_{cY}^{n}(\D,\D)\}_{n \geq 1}, \circ, \mathbbm{1})$ considered in Theorem \ref{thm-operad}. Hence, a result of Gerstenhaber and Voronov \cite{gers-voro} shows that the element $\pi$ induces a degree $0$ associative cup product $\smile$ on the graded $\C$-module $\bigoplus_{n =1 }^{\infty} C_{cY}^{n}(\D,\D)$, given by
\begin{align*}
    f \smile g := (-1)^{mn+1} ~ \! (\pi \circ_{2} g) \circ_{1} f, \text{ for } f \in C_{cY}^{m}(\D,\D), ~ \! g \in C_{cY}^{n}(\D,\D).
\end{align*}
Explicitly, for any $y \in Y_{m+n}$ and $a_1, \ldots, a_{m+n} \in \D$,
\begin{align}\label{cup-pro}
    &(f\smile g)_{\la_1,\ldots,\la_{m+n-1}}
(y;a_1,\ldots,a_{m+n}) \nonumber
=(-1)^{mn+1} ~ \!
\pi_{\la_1+\cdots+\la_m}
\big(
R_0^{\,2;2,n}R_0^{\,n+1;1,m}y;
\nonumber\\
&\qquad \qquad
f_{\la_1,\ldots,\la_{m-1}}
\big(R_1^{\,n+1;1,m}y;a_1,\ldots,a_m\big),
g_{\la_{m+1},\ldots,\la_{m+n-1}}
\big(
R_2^{\,2;2,n}R_0^{\,n+1;1,m}y;
a_{m+1},\ldots,a_{m+n}
\big)
\big).
\end{align}
\medskip
Further, the cup product satisfies
\begin{align*}
    \delta_{cDiass}(f \smile g) = \delta_{cDiass}(f) \smile g + (-1)^{m} ~ \! f \smile \delta_{cDiass}(g).
\end{align*}
This shows that the cup product $\smile$ induces a graded associative product (denoted by the same notation) on the graded cohomology space $H^\bullet_{cDiass} (\D, \D)$. The induced product $\smile$ turns out to be graded commutative, and the identity (\ref{leibn-rule}) holds with respect to the induced degree $-1$ bracket $[~,~]$. Summarizing this, we get the following result.

\begin{theorem}\label{thm-gers}
Let $(\D,\dashv_{\la},\vdash_{\la})$ be a diassociative conformal algebra. Then the cup product (\ref{cup-pro}) and the degree $-1$ bracket $[\, ,\, ]$ induce well-defined operations on the graded cohomology space $H_{cDiass}^{\bullet}(\D,\D)
$. Moreover, with the induced operations, the triple $\left(H_{cDiass}^{\bullet}(\D,\D), ~ \! \smile ~ \! ,\,[\, ,\, ]\right)$ is a Gerstenhaber algebra.
\end{theorem}

\medskip

\noindent {\bf Cohomology with coefficients in a {representation}.} \begin{definition}
Let $(\D,\dashv_{\la},\vdash_{\la})$ be a diassociative conformal algebra. Then a  {\bf representation} of $\D$ is a $\C[\partial]$-module $\M$ together with four conformal sesquilinear $\la$-actions
\[
\dashv_{\la},\,\vdash_{\la}:\D\otimes\M\longrightarrow\M[\la] ~ \text{ and } ~
\dashv_{\la},\,\vdash_{\la}:\M\otimes\D\longrightarrow\M[\la],
\]
satisfying a list of fifteen compatibility identities obtained from the defining identities \eqref{equation 1}--\eqref{equation 5} of the diassociative conformal algebra $\D$ by replacing exactly one of the variables by an element of $\M$. 
\end{definition}


Any diassociative conformal algebra $(\D, \dashv_\la, \vdash_\la)$ admits a canonical representation on $\D$ itself, called the \emph{adjoint representation}, in which the four $\la$-actions are given by the $\la$-products $\dashv_{\la}$ and $\vdash_{\la}$.

\medskip

Let $(\D, \dashv_{\la}, \vdash_{\la})$ be a diassociative conformal algebra and $\M$ be a representation of $\D$. For each integer $n \geq 0$, we define the space of $n$-cochains by
\begin{align*}
    C_{cDiass}^{n}(\D,\M) := 
    \begin{cases}
       \M/{\p \M},   &  n =0,\\
        C_{cY}^{n}(\D,\M),  & n \geq 1,
    \end{cases}
\end{align*}
where $C_{cY}^{n}(\D,\M)$ is the space of all conformal sesquilinear maps from $\mathbb{C}[Y_n]\otimes \D^{\otimes n}$ to $\M[\la_{1}, \dots ,\la_{n-1}]$. 
We define a map $\delta_{cDiass} : C_{cDiass}^{n}(\D,\M) \to  C_{cDiass}^{n+1}(\D,\M)$ by
\begin{align}\label{coboundary-m1}
    \delta_{cDiass}(u+\p\M)(a) := \left(a\dashv_{-\p-\la}u-u\vdash_{\la}a\right) \big|_{\la=0}, \text{ for } a \in \D \text{ and } u \in \M,
\end{align}
\begin{align}\label{coboundary-m2}
   {\delta_{cDiass}(f)}_{\la_1, \ldots, \la_n}&(y;a_1, \ldots , a_{n+1}) :=~  
a_{1} (\bullet_{0}^{y})_{ \la_{1}}f_{ \la_{2}, \ldots, \la_{n}}(d_{0}y; a_{2}, \ldots , a_{n+1}) \nonumber \\&+ \sum_{i =1}^{n}(-1)^{i} ~ \! f_{ \la_{1}, \ldots , \la_{i-1}, \la_{i}+ \la_{i+1}, \la_{i+2}, \ldots , \la_{n}}(d_{i}y; a_{1}, \ldots, a_{i-1},a_{i}(\bullet_{i}^{y})_{\la_{i}} a_{i+1}, \ldots, a_{n+1}) \nonumber \\&
+(-1)^{n+1} ~ \! f_{\la_{1}, \ldots, \la_{n-1}}(d_{n+1}y; a_{1}, \ldots, a_{n})(\bullet_{n+1}^{y})_{\la_{1}+ \cdots+ \la_{n} }a_{n+1},
\end{align}
for all $f \in C_{cDiass}^{n \geq 1}(\D,\M), ~ y \in Y_{n+1} \text{ and } a_1, \ldots , a_{n+1} \in \D$. Then we have the following result.
\begin{proposition}
   The map $\delta_{cDiass} : C_{cDiass}^{n}(\D,\M) \to  C_{cDiass}^{n+1}(\D,\M)$ satisfies the identity $\delta_{cDiass}^{2} = 0$. Consequently, the pair $\{C_{cDiass}^{\bullet}(\D,\M), \delta_{cDiass}\}$ forms a cochain complex, called the cochain complex of the diassociative conformal algebra $\D$ with coefficients in the representation $\M$.
\end{proposition}

\begin{proof}
Since $(\D,\dashv_{\la},\vdash_{\la})$ is a diassociative conformal algebra and $\M$ is a representation of it, the direct sum $\M \oplus \D$ carries a diassociative conformal algebra structure whose $\la$-products (also denoted by the same notations) are explicitly given by
\begin{align*}
    ( u, a) \dashv_\la ( v, b) = (  a \dashv_\la v + u \dashv_\la b ~ \!  , ~ \!  a \dashv_\la b ) ~ \text{ and } ~ ( u, a) \vdash_\la (v, b) = (  a \vdash_\la v + u \vdash_\la b ~ \! , ~ \! a \vdash_\la b),
\end{align*}
for $( u, a), ( v, b) \in \M \oplus \D$. Hence, one may consider the cochain complex $\{ C_{cDiass}^{\bullet}(\M \oplus \D, \M \oplus \D), \delta_{cDiass} \}$ of this diassociative conformal algebra. For each $n \geq 0$, we observe that
\begin{align*}
     C_{cDiass}^{n}(\D,\M) \subseteq C_{cDiass}^{n}(\M \oplus \D, \M \oplus \D) ~~ \text{ and } ~~ \delta_{cDiass}  (C_{cDiass}^{n}(\D,\M) ) \subseteq C_{cDiass}^{n+1}(\D,\M).
\end{align*}
Moreover, the restriction of the coboundary map $\delta_{cDiass}$ on the subspace $C_{cDiass}^{n}(\D,\M)$ coincides with the map (that has been denoted by the same notation) given in (\ref{coboundary-m1}), (\ref{coboundary-m2}). Hence $\{C_{cDiass}^{\bullet}(\D,\M), \delta_{cDiass}\}$ is a cochain complex.
\end{proof}

We denote the space of all $n$-cocycles and $n$-coboundaries in the complex $\{C_{cDiass}^{\bullet}(\D,\M), \delta_{cDiass}\}$ by the notations $Z^n(\D,\M)$ and $B^n(\D,\M)$, respectively. Finally, the cohomology groups of this complex are denoted by $H_{cDiass}^{n}(\D,\M)$, for $n \geq 0$, and called the cohomology of the diassociative conformal algebra $\D$ with coefficients in the representation $\M$.

\medskip

The zeroth cohomology group $H_{cDiass}^{0}(\D,\M)$ is simply given by
\begin{align*}
H_{cDiass}^{0}(\D,\M)= \{u \in \M \mid a \dashv_{-\partial} u = u \vdash_{(0)} a, \text{ for } a \in \D\} / \partial \M.
\end{align*}
To understand the first cohomology group, we first consider the following notion.
\begin{definition}
Let $(\D, \dashv_{\la}, \vdash_{\la})$ be a diassociative conformal algebra and $\M$ be a representation of $\D$. Then a {\bf derivation} of $\D$ with values in $\M$ is a $\C[\p]$-linear map $d: \D \rightarrow \M$ satisfying
\begin{align*}
d(a\dashv_{\la}b)=d(a)\dashv_{\la}b + a\dashv_{\la}d(b) ~ \text{ and } ~
d(a\vdash_{\la}b)=d(a)\vdash_{\la}b + a\vdash_{\la}d(b),  \text{ for } a,b\in \D.
\end{align*}
\end{definition}
We denote by $\operatorname{Der}(\D,\M)$ the $\C[\p]$-module of all derivations of $\D$ with values in $\M$. For each $u \in \M$, we set a map $\operatorname{ad}_{u}: \D \rightarrow \M$
by
\begin{align*}
\operatorname{ad}_{u}(a) := a\dashv_{-\p}u - u\vdash_{(0)}a, \text{ for } a\in \D.
\end{align*}
Then $\mathrm{ad}_u$ is a derivation of $\D$ with values in $\M$.
Such derivations are called \emph{inner derivations}. The set of all inner derivations is denoted by $\operatorname{Inn}(\D,\M)$. Then it is easy to see that
\begin{align*}
H_{cDiass}^{1}(\D,\M)
=\operatorname{Der}(\D,\M)/\operatorname{Inn}(\D,\M).
\end{align*}


\subsection*{Relation with annihilation diassociative algebras}

In \cite{bakalov-kac-voronov}, the authors established a relationship between the cohomology of a Lie conformal algebra and that of its annihilation Lie algebra. Here we prove an analogous result for diassociative conformal algebras. As an application of our result, we describe the cohomology of a current diassociative conformal algebra in terms of the cohomology of a polynomial diassociative algebra.

Let $(\D, \dashv_\la, \vdash_\la)$ be a diassociative conformal algebra and $\M$ be a representation of it. Consider the quotient space $V (\M) = \M[t^{\pm 1}]\big/{(\p + \p_t)~ \!\M[t^{\pm 1}]}$. For $u\in\M$ and $m\in\mathbb{Z}$, we denote by $u_m$ the image of $ut^m$ in $V (\M)$. The four conformal sesquilinear $\la$-actions of $\D$ on $\M$ induce four actions of the diassociative algebra $\mathrm{Diass} (\D) = \D[t^{\pm1}]\big/(\partial+\partial_t) ~ \! \D[t^{\pm1}]$ on $V (\M)$. These actions are defined by formulas analogous to \eqref{equation 11}, with exactly one of the entries replaced by an element of $V(\M)$. With these actions, $V(\M)$ becomes a representation of the diassociative algebra $\mathrm{Diass} (\D)$. Let $\M_-$ be the subspace of $V(\M)$ spanned by the elements $u_m$ with $u\in\M$ and $m\geq0$. It follows immediately from the above formulas that $\M_-$ is stable under the four actions of the annihilation diassociative algebra $\D_-$ considered in Remark~\ref{anni-di}. For any $n \geq 0$, there is a $\p$-action on the space 
\begin{align*}
    C^n_{Diass} (\D_{-},\M_{-}) = \big\{ f : \C [Y_n] \otimes \D_{-}^{\otimes n} \rightarrow \M_{-} ~ \! \big| ~ \! \substack{ f \text{ is } \C\text{-linear and } f (y;{a_1}_{m_1}, \ldots, {a_n}_{m_n} ) |_{t=0}= 0, ~ \\  \text{for all but finitely many } m_1, \ldots, m_n \geq 0 } \big\}.
\end{align*}
Explicitly, for any $f \in C^{n \geq 1} _{Diass}(\D_{-},\M_{-})$ and ${a_1}_{m_1}, \ldots, {a_n}_{m_n} \in \D_-$,
\begin{align}\label{delta-action}
    (\p f)(y;{a_1}_{m_1}, \ldots, {a_n}_{m_n})= -\p_t(f(y;{a_1}_{m_1}, \ldots, {a_n}_{m_n})) - \sum_{i=1}^{n}f(y;{a_1}_{m_1}, \ldots, {(\p a_i)}_{m_i}, \ldots {a_n}_{m_n})
\end{align}
and $\p u_m = -m u_{m-1}$, for any $u_m \in C^0_{Diass} (\D_{-},\M_{-}) \cong \M_{-}$. Let 
\begin{align*}
C^n_{Diass}(\D_{-},\M_{-})^{\partial}=\{ f \in C^n_{Diass} (\D_{-},\M_{-}) \, \big | \, \p f =0 \}
\end{align*}
be the $\p$-invariant subspace of $C^n_{Diass} (\D_{-},\M_{-})$. Then it turns out that $C^\bullet_{Diass}(\D_{-},\M_{-})^{\partial}$ is a cochain complex with the usual differential of the diassociative algebra $\D_{-}$ with coefficients in $\M_{-}$ \cite{Frabetti}. Then we have the following result.
\begin{theorem}\label{th-39}
Let $\mathcal{D}$ be a diassociative conformal algebra and $\M$ be
a representation of $\mathcal{D}$ which is free as a
$\mathbb{C}[\partial]$-module. Then there is a natural isomorphism of
cochain complexes
\begin{align*}
C^\bullet_{cDiass}(\mathcal{D},\M)
\cong
C^\bullet_{Diass}
(\D_{-},\M_{-})^{\partial}.
\end{align*}
\end{theorem}

\begin{proof}
Since $\M$ is free as a $\mathbb{C}[\partial]$-module, there exists a
complex vector space $U$ such that $\mathcal{M} = \mathbb{C} [\p] \otimes U$. Consequently, the annihilation module $\M_-$ is naturally isomorphic to
$U[t]$. Let $f \in C^n_{Diass} (\D_{-},\M_{-})$ be arbitrary. We consider the map $\beta : \C[Y_n] \otimes \D^{\otimes n} \to U \llbracket \la_1, \ldots, \la_n,t \rrbracket$ by 
\begin{align}\label{beta-f}
    \beta_{\la_1, \ldots, \la_n, t }(y; a_1, \ldots, a_n)= \sum _{m_1, \ldots, m_n \geq 0} \frac{\la_1^{m_1}}{m_1 !} \cdots \frac{\la_n^{m_n}}{m_n !} ~ \! f(y;{a_1}_{m_1}, \ldots, {a_n}_{m_n} ),
\end{align}
for $y \in Y_n$ and $a_1, \ldots, a_n \in \D$. It turns out that the $\p$-action on $f$ defined in \eqref{delta-action} is equivalent to the action of $(-\p_t + \sum \limits _{i=1}^{n} \la_i)$ on $\beta_{\la_1, \ldots, \la_n,t}$. Hence, $f$ is $\p$-invariant if and only if $(-\p_t + \sum \limits _{i=1}^{n} \la_i)\beta_{\la_1, \ldots, \la_n,t}=0$. This is a first-order differential equation in the variable $t$, whose solution is given by 
\begin{align}\label{beta}
    \beta_{\la_1, \ldots, \la_n,t}(y; a_1, \ldots, a_n) = \gamma_{\la_1, \ldots, \la_n}(y; a_1, \ldots, a_n) ~ \! e^{t \sum \limits _{i=1}^{n}\la_i},
\end{align}
where $\gamma_{\la_1, \ldots, \la_n} = \beta_{\la_1, \ldots, \la_n,t}|_{t=0}$. By identifying $U$ with $1 \otimes U \subset \M$, we may define $\overline{\gamma} \in C^n_{cDiass}(\mathcal{D},\M)$ by 
\begin{align*}
    \overline{\gamma}_{\la_1, \ldots, \la_{n-1}}(y; a_1, \ldots, a_n) = \gamma_{\la_1, \ldots, \la_{n-1}, -\p - \sum \limits _{i=1}^{n-1}\la_i}(y; a_1, \ldots, a_n).
\end{align*}
It is easy to check that $f \mapsto \overline{\gamma} $ is a cochain map from $C^\bullet_{Diass}
(\D_{-},\M_{-})^{\partial}$ to $C^\bullet_{cDiass}(\mathcal{D},\M)$.

Conversely, let $ \overline{\gamma} \in C^n_{cDiass}(\mathcal{D},\M)$. We define a map $\gamma : \C[Y_n] \otimes \D^{\otimes n} \to \M[\la_1, \ldots, \la_n]$ by 
\begin{align*}
    \gamma_{\la_1, \ldots, \la_n}(y; a_1, \ldots, a_n) = \overline{\gamma}_{\la_1 - \frac{\la+\p}{n}, \ldots, \la_{n-1} - \frac{\la+\p}{n}}(y; a_1, \ldots, a_n), ~~\text{ where } \la = \sum_{i=1}^{n} \la_i.
\end{align*}
Subsequently, we define $ \beta_{\la_1, \ldots, \la_n,t} (a_1, \ldots, a_n) \in \M \llbracket \la_1, \ldots, \la_n, t \rrbracket$ simply by using \eqref{beta}. Since $\M = \C [\partial] \otimes U$, we may substitute $\partial$ by $-\p_t$ in the expression of  $ \beta_{\la_1, \ldots, \la_n,t} (a_1, \ldots, a_n)$. Then it lies in $U \llbracket \la_1, \ldots, \la_n, t \rrbracket$. It turns out that $ \beta_{\la_1, \ldots, \la_n,t}$ is $(-\p_t + \sum \limits _{i=1}^{n} \la_i)$-invariant. Hence the map $f \in C^n_{Diass}
(\D_{-},\M_{-})$ defined in \eqref{beta-f}, is $\p$-invariant. Finally, the correspondence $f \leftrightarrow \overline{\gamma}$ is an isomorphism, which concludes the proof.
\end{proof}

In the following result, we use the above theorem to obtain an explicit description of the cohomology of a current diassociative conformal algebra. More precisely, we have the following.
 
\begin{proposition}\label{prop-cur}
Let $D$ be a finite-dimensional diassociative algebra and $M$ be a
representation of $D$. 

(i) We set $\operatorname{Cur}(M)=\mathbb{C}[\partial]\otimes M$ and define four conformal sesquilinear $\la$-actions
\begin{align*}
\dashv_\la,\ \vdash_\la:
\operatorname{Cur}(D)\otimes \operatorname{Cur}(M)
\longrightarrow \operatorname{Cur}(M)[\la] ~\text{ and } ~\dashv_\la,\ \vdash_\la:
\operatorname{Cur}(M)\otimes \operatorname{Cur}(D)
\longrightarrow \operatorname{Cur}(M)[\la],
\end{align*}
by
\[
\begin{cases}
(1\otimes a)\dashv_\la(1\otimes u)
 =1\otimes(a\dashv u),\\
(1\otimes a)\vdash_\la(1\otimes u)
 =1\otimes(a\vdash u),
\end{cases}
\qquad
\begin{cases}
(1\otimes u)\dashv_\la(1\otimes a)
 =1\otimes(u\dashv a),\\
(1\otimes u)\vdash_\la(1\otimes a)
 =1\otimes(u\vdash a),
\end{cases}
\]
for $a\in D$ and $u\in M$. With the above operations, $\operatorname{Cur}(M)$ is a representation
of the current diassociative conformal algebra $\operatorname{Cur}(D)$
given in Example~\ref{cur-d}. 

\medskip

(ii) Let $\widetilde{H}^\bullet_{{Diass}}\bigl( D[t], M \bigr)$ be the cohomology of the subcomplex of the cochain complex of the diassociative algebra $D[t]$ with coefficients in the representation $M$, obtained by restricting on cochains with finite support. Then we have
\[
H^\bullet_{{cDiass}}
\bigl(\operatorname{Cur}(D),\operatorname{Cur}(M)\bigr)
\cong
\widetilde{H}^\bullet_{{Diass}}\bigl(D[t],M\bigr).
\]
\end{proposition}
\begin{proof}
The first part is straightforward. For the second part, we observe that both $\mathrm{Cur}(D) = \C [\p] \otimes D$ and $\operatorname{Cur}(M) = \C [\p] \otimes M$ are free $\C[\p]$-modules. Hence, one have $\mathrm{Cur}(D)_- = D[t]$ and $\mathrm{Cur}(M)_- = M[t]$. Thus by Theorem \ref{th-39}, 
\begin{align}\label{iso-1}
    H^\bullet_{{cDiass}}
\bigl(\operatorname{Cur}(D),\operatorname{Cur}(M)\bigr)
\cong
H^\bullet_{{Diass}}\bigl(D[t],M[t]\bigr)^\p.
\end{align}
Next, for any $f \in C^n_{{Diass}}\bigl(D[t],M[t]\bigr)^\p$, we define an element $\overline{\gamma} \in C^n_{{Diass}}\bigl(D[t],M\bigr)$ by evaluating $f$ at $t=0$. This construction $f \mapsto \overline{\gamma}$ yields a cochain map from $C^\bullet_{{Diass}}\bigl(D[t],M[t]\bigr)^\p$ to $C^\bullet_{{Diass}}\bigl(D[t],M\bigr)$. 

Conversely, for any $\overline{\gamma} \in C^n_{{Diass}}\bigl(D[t],M\bigr) $, we define a map $\gamma : \C[Y_n] \otimes (\operatorname{Cur}(D))^{\otimes n} \to M [\la_1, \ldots, \la_n]$ by 
\begin{align*}
    \gamma_{\la_1, \ldots, \la_n }(y;1 \otimes {a_1}, \ldots, 1 \otimes {a_n})= \sum _{m_1, \ldots, m_n \geq 0} \frac{\la_1^{m_1}}{m_1 !} \cdots \frac{\la_n^{m_n}}{m_n !} ~ \!  \overline{\gamma}(y;{a_1}t^{m_1}, \ldots, {a_n}t^{m_n} ).
\end{align*}
Subsequently, we set a map $\beta : \C[Y_n] \otimes (\operatorname{Cur}(D))^{\otimes n} \to M \llbracket \la_1, \ldots, \la_n,t \rrbracket$ by 
\begin{align*}
     \beta_{\la_1, \ldots, \la_n,t}(y;1 \otimes {a_1}, \ldots, 1 \otimes {a_n} ) = \gamma_{\la_1, \ldots, \la_n}(y; 1 \otimes {a_1}, \ldots, 1 \otimes {a_n} ) ~ \! e^{t \sum \limits _{i=1}^{n}\la_i}.
\end{align*}
Finally, we consider a map $f \in C^n_{{Diass}}\bigl(D[t],M[t]\bigr)^\p$ by 
\begin{align*}
     \beta_{\la_1, \ldots, \la_n,t }(y;1 \otimes {a_1}, \ldots, 1 \otimes {a_n})= \sum _{m_1, \ldots, m_n \geq 0} \frac{\la_1^{m_1}}{m_1 !} \cdots \frac{\la_n^{m_n}}{m_n !} ~ \!  f(y;{a_1}t^{m_1}, \ldots, {a_n}t^{m_n} ).
\end{align*}
It is easy to see that the correspondence $f \leftrightarrow \overline{\gamma}$ is an isomorphism. Hence, we have $H^\bullet_{{Diass}}
\bigl(D[t],M[t]\bigr)^\p \cong \widetilde{H}^\bullet_{{Diass}}
\bigl(D[t],M\bigr)$. Combining this with the isomorphism \eqref{iso-1}, we get the required result. 
\end{proof}

\section{Applications of cohomology}\label{sec4}

In this section, we study extensions and deformations of diassociative conformal algebras using the cohomology introduced in the previous section. In particular, we show that the set of equivalence classes of extensions of a diassociative conformal algebra $\D$ by a representation $\M$ has a bijection with the second cohomology group $H^2_{cDiass} (\D, \M)$. Subsequently, we consider first-order deformations of a diassociative conformal algebra $\D$ and show that the set of isomorphism classes of such deformations can be classified by $H^2_{cDiass} (\D, \D)$.

 \medskip
   
 \noindent {\bf Extensions of diassociative conformal algebras.}

Let $(\D,\dashv_{\la},\vdash_{\la})$ be any diassociative conformal algebra, and $\mathcal{A}$ be an abelian diassociative conformal algebra, that is, both the $\la$-products on $\mathcal{A}$ are trivial.
A {\bf $\mathbb{C}[\partial]$-split extension} of $\D$ by $\mathcal{A}$ is a $\mathbb{C}[\partial]$-split short exact sequence of diassociative conformal algebras
\begin{align}\label{121}
0\longrightarrow
\mathcal{A}
\xrightarrow{\ i \ }
\E
\xrightarrow{\ \pi\ }
\D
\longrightarrow 0.
\end{align}
A $\mathbb{C}[\partial]$-split extension \eqref{121} naturally induces a {representation} of the diassociative conformal algebra $\D$ on the $\mathbb{C}[\partial]$-module $\mathcal{A}$. More precisely, for any $a \in \D$ and $u \in \mathcal{A}$, the conformal sesquilinear $\la$-actions are given by
\begin{align*}
   & a \dashv_\la u :=  s(a) \dashv_\la i(u)   , \qquad  a \vdash_\la u :=  s(a) \vdash_\la i(u), \qquad  u \dashv_\la a := i(u) \dashv_\la s(a)  \\
   & \qquad \qquad \qquad \qquad \quad \text{ and } \quad u \vdash_\la a := i(u) \vdash_\la s(a),
\end{align*}
where $s : \D \rightarrow \E$ is any $\mathbb{C}[\partial]$-linear section of the map $\pi$. This is called the {\em induced representation} of the above extension.


\begin{definition}
Let $(\D,\dashv_{\la},\vdash_{\la})$ be a diassociative conformal algebra and $\M$ be a representation of it. Then a {\bf $\mathbb{C}[\partial]$-split extension} of $\D$ by $\M$ is a $\mathbb{C}[\partial]$-split short exact sequence of diassociative conformal algebras
\[
0\longrightarrow
\M
\xrightarrow{\ i \ }
\E
\xrightarrow{\ \pi\ }
\D
\longrightarrow 0
\]
(where $\M$ is regarded as an abelian diassociative conformal algebra) for which the induced representation of $\D$ on $\M$ agrees with the prescribed one.
\end{definition}
Two $\C[\p]$-split extensions $0 \longrightarrow \M
\xrightarrow{\ i \ } \E
\xrightarrow{\ \pi \ } \D
\longrightarrow 0$ and $
0 \longrightarrow \M
\xrightarrow{\ i' \ } \E'
\xrightarrow{\ \pi' \ } \D
\longrightarrow 0
$ are said to be {\bf equivalent} if there exists an isomorphism of diassociative conformal algebras $\varphi:  \E\longrightarrow\E'$ such that the following diagram commutes: 

\[
\begin{tikzcd}
0 \arrow[r] &
\M \arrow[r,"i"] \arrow[d,equals] &
\E \arrow[r,"\pi"] \arrow[d,"\varphi"] &
\D \arrow[r] \arrow[d,equals] &
0\\
0 \arrow[r] &
\M \arrow[r,"i'"'] &
\E' \arrow[r,"\pi'"'] &
\D \arrow[r] &
0.
\end{tikzcd}
\]

\medskip

\noindent We denote by $\operatorname{Ext}(\D,\M)$ the set of equivalence classes of $\C[\p]$-split extensions of the diassociative conformal algebra $\D$ by the representation $\M$. Then we have the following result.

\begin{theorem}\label{thm-ext}
Let $(\D,\dashv_{\la},\vdash_{\la})$ be a diassociative conformal algebra and $\M$ be a representation of it. Then there is a natural bijection between $\operatorname{Ext}(\D,\M)$ and $H_{cDiass}^{2}(\D,\M)$.
\end{theorem}

\begin{proof}
Let $0\longrightarrow
\M
\xrightarrow{\ i \ }
\E
\xrightarrow{\ \pi\ }
\D
\longrightarrow 0$ be a $\C[\p]$-split extension of the diassociative conformal algebra $\D$ by the representation $\M$. Since the extension is $\C[\p]$-split, we may identify $\E$ with the direct sum $\M \oplus \D$ as a $\C[\p]$-module. Under this identification, we have $i(u)=(u,0), ~ \! \pi(u,a)=a$ and $\p(u ,a)=(\p u,\p a)$, for all $u \in \M$ and $ a \in \D$. Since $\pi$ is a homomorphism of diassociative conformal algebras,
\begin{align*}
\pi ((0,a)\dashv_{\la}(0,b))=a \dashv_{\la} b ~~~ \text{ and } ~~~
\pi((0,a) \vdash_{\la} (0,b))=a \vdash_{\la} b,  \text{ for all } a, b \in \D.
\end{align*}
Therefore, we obtain a conformal sesquilinear map $f : \C[Y_2]\otimes \D^{\otimes 2 } \to \M[\la]$ such that
\begin{align*}
(0, a) \dashv_{\la}(0,b)
=
\big(f_{\la} (\tikz[baseline=-0.5ex,scale=0.20]{
  \draw (-1,1)--(0,0);
  \draw (1,1)--(0,0);
  \draw (0,0)--(0,-1);
  \draw (0.5,0.5)--(0,1);
};a, b ), \, a \dashv_{\la} b \big) ~~~ \text{ and } ~~~ (0,a) \vdash_{\la} (0,b) = \big(f_{\la} (\tikz[baseline=-0.5ex,scale=0.20]{
  \draw (-1,1)--(0,0);
  \draw (1,1)--(0,0);
  \draw (0,0)--(0,-1);
  \draw (-0.5,0.5)--(0,1);
};  a, b ), \, a \vdash_{\la} b \big).
\end{align*}
Since the representation on $\M$ induced by the extension coincides with the given representation and the $\la$-products on $\M$ are trivial, it follows that the $\la$-products of $\E$ are given by
\begin{align}
(u, a) \dashv_{\la} (v ,b) &=
\big(
a \dashv_{\la} v
+
u \dashv_{\la} b
+
f_{\la} (\tikz[baseline=-0.5ex,scale=0.20]{
  \draw (-1,1)--(0,0);
  \draw (1,1)--(0,0);
  \draw (0,0)--(0,-1);
  \draw (0.5,0.5)--(0,1);
};a, b ),
\,
a \dashv_{\la} b
\big), \label{p-1}\\
(u, a)\vdash_{\la}(v, b)
&=
(
a \vdash_{\la} v
+
u \vdash_{\la} b
+
f_{\la} (\tikz[baseline=-0.5ex,scale=0.20]{
  \draw (-1,1)--(0,0);
  \draw (1,1)--(0,0);
  \draw (0,0)--(0,-1);
  \draw (-0.5,0.5)--(0,1);
};a, b ),
\,
a \vdash_{\la} b
) \label{p-2}.
\end{align}
One may easily verify that the $\la$-products (\ref{p-1}) and (\ref{p-2}) define a diassociative conformal algebra structure on $\E$ if and only if the conformal sesquilinear map $f$ satisfies the condition $ \delta_{cDiass}(f)=0$. Equivalently, $f \in Z_{cDiass}^{2}(\D,\M)$ is a $2$-cocycle.

Next, suppose $0\longrightarrow
\M
\xrightarrow{\ i \ }
\E
\xrightarrow{\ \pi\ }
\D
\longrightarrow 0$ and $0\longrightarrow
\M
\xrightarrow{\ i' \ }
\E'
\xrightarrow{\ \pi' \ }
\D
\longrightarrow 0$
are equivalent $\C[\partial]$-split extensions, where the diassociative conformal algebra structures $(\E, \dashv_{\la}, \vdash_{\la})$ and $(\E', \dashv_{\la}',\vdash_{\la}')$ are determined by the $2$-cocycles $f$ and $f'$, respectively. Then there exists an isomorphism of diassociative conformal algebras $\varphi : \E \cong \M \oplus \D \longrightarrow \M \oplus \D \cong \E'$ such that 
\begin{align}
\pi\circ \varphi (u, a)= a ~ \text{ and } ~
\varphi (u,0)=(u,0),  \text{ for } a \in \D,  u \in \M.  \label{eq4}
\end{align}
Hence, there exists a $\C [\partial]$-linear map $g : \D \longrightarrow \M$ such that $\varphi (u, a) = (u + g (a), a)$, for all $(u, a) \in \E \cong \M \oplus \D$.
Since $\varphi$ is a homomorphism of diassociative conformal algebras, we have
\begin{align*}
\varphi ((0,a)\dashv_{\la}(0,b) )
=
\varphi  (0,a)\dashv_{\la}' \varphi (0,b)~
\text{ and }~
\varphi ((0,a)\vdash_{\la}(0,b) )
=
\varphi (0,a)\vdash_{\la}' \varphi (0,b ),
\end{align*}
for all $a, b \in \D$. These conditions yield $f'-f=\delta_{cDiass}(g)$, which shows that $f'$ and $f$ represent the same cohomology class in $H_{cDiass}^{2}(\D,\M)$.

\medskip

Conversely, let $f\in Z_{cDiass}^{2}(\D,\M)$ be a $2$-cocycle. Then the $\la$-products defined in \eqref{p-1} and \eqref{p-2} make $\M\oplus\D$ into a diassociative conformal algebra structure. This defines a $\mathbb{C}[\partial]$-split extension of the diassociative conformal algebra $\D$ by the representation $\M$. Next, suppose $f, f' \in  Z_{cDiass}^{2}(\D,\M)$ are cohomologous, say $ f'- f = \delta_{cDiass}(g)$ for some $g\in C_{cDiass}^{1}(\D,\M)$. Then it is easy to see that the map
\begin{align*}
\varphi: \M\oplus\D\longrightarrow\M\oplus\D  \text{ given by }
\varphi(u, a) :=(u+g(a),a), \text{ for } (u, a) \in \M \oplus \D,
\end{align*}
is an isomorphism between the $\mathbb{C}[\partial]$-split extensions corresponding to the 2-cocycles $f$ and $f'$. Hence cohomologous $2$-cocycles determine equivalent $\mathbb{C}[\partial]$-split extensions. 

Finally, the two correspondences above are inverses of each other.
\end{proof}

\medskip
   
 \noindent {\bf Deformations of diassociative conformal algebras.} Let $A = \C [\epsilon] / (\epsilon^2 )$ be the algebra of dual numbers. For any diassociative conformal algebra $\D$, we consider the tensor product $\D \otimes A$, which is obviously a $\C [\partial]$-module where the $\partial$-action is given on the first factor. Then a {\bf first-order deformation} of $\D$ is given by a diassociative conformal algebra structure on $\D \otimes A$ over $A$, so that the projection map
 \begin{align*}
    \pi: \D \otimes A \rightarrow \D ~\text{ defined by }~ \pi (a \otimes p (\epsilon) ) = p (0) a
 \end{align*}
 is a homomorphism of diassociative conformal algebras. Two first-order deformations are said to be {\bf isomorphic} if there exists an isomorphism $\varphi: \D \otimes A \rightarrow \D \otimes A$ of the corresponding diassociative conformal algebras over $A$ commuting with the corresponding projection maps.

 It follows that a first-order deformation is determined by a $\C[\partial]$-split extension
 \begin{align*}0\longrightarrow
\D
\xrightarrow{\ i \ }
\D \otimes A
\xrightarrow{\ \pi\ }
\D
\longrightarrow 0
\end{align*}
of the diassociative conformal algebra $\D$ by the adjoint representation, where $i (a) = a \otimes \epsilon$, for $a \in \D$. Moreover, any two first-order deformations are isomorphic if and only if the corresponding $\C[\partial]$-split extensions are equivalent. Hence, by Theorem \ref{thm-ext}, we get the following result.

\begin{thm}\label{thm-def}
    Let $\D$ be a diassociative conformal algebra. Then the set of all isomorphism classes of first-order deformations of $\D$ has a bijection with the second cohomology group $H^2_{cDiass} (\D, \D).$
\end{thm}

\section{Averaging operators on associative conformal algebras}\label{sec5}

In this section, we introduce averaging operators on associative conformal algebras. A more general notion could be a relative averaging operator on an associative conformal algebra with respect to a conformal bimodule. We investigate the close relationship between relative averaging operators and diassociative conformal algebras. Finally, we find a Maurer-Cartan characterization of (relative) averaging operators and define cohomology of such operators.

\medskip

Let $\mathcal{A}$ be an associative conformal algebra. Then a $\mathbb{C}[\partial]$-linear map $P:\mathcal{A} \rightarrow \mathcal{A}$ is called an {\bf averaging operator} on $\mathcal{A}$ if
\begin{align*}
P(a)_{\la}P(b)
=
P\bigl(P(a)_{\la}b\bigr)
=
P\bigl(a_{\la}P(b)\bigr),
\text{ for }a,b\in\mathcal{A}.
\end{align*}


We will now consider a generalization of averaging operators closely related to diassociative conformal algebras.

\begin{definition}
Let $\mathcal{A}$ be an associative conformal algebra and $\mathcal{M}$ be a conformal $\mathcal{A}$-bimodule. Then a {\bf relative averaging operator} on $\mathcal{A}$ with respect to the conformal $\mathcal{A}$-bimodule $\mathcal{M}$ is a $\mathbb{C}[\partial]$-linear map $P:\mathcal{M}\rightarrow\mathcal{A}$ that satisfies
\begin{align*}
P(u)_{\la}P(v)
=
P\bigl(P(u)_{\la} v\bigr)
=
P\bigl(u_{\la}P(v)\bigr), \text{ for } u, v \in\mathcal{M}.
\end{align*}
A datum $(\mathcal{A},\mathcal{M}, P)$ consisting of an associative conformal algebra $\mathcal{A}$, a conformal $\mathcal{A}$-bimodule $\M$, and a relative averaging operator $P: \M \rightarrow \mathcal{A}$ is referred to as a \textbf{relative averaging conformal algebra}.
\end{definition}

\begin{definition}\label{dmor}
Let $(\mathcal{A},\mathcal{M},P)$ and
$(\mathcal{A}',\mathcal{M}',P')$ be two relative averaging conformal algebras. Then a {\bf homomorphism} of relative averaging conformal algebras
from the former one to the latter one is a pair $(f,g)$, where
$f:\mathcal{A}\to\mathcal{A}'$ is a homomorphism of associative conformal algebras and
$g:\mathcal{M}\to\mathcal{M}'$ is a $\mathbb{C}[\partial]$-linear map satisfying
\begin{align*}
g(a_{\la}u)=f(a)_{\la}g(u), \quad 
g(u_{\la}a)=g(u)_{\la}f(a) ~~~ \text{ and }
~~~ f\circ P=P'\circ g, \text{ for } a\in\mathcal{A},\,u\in\mathcal{M}.
\end{align*}
\end{definition}

It is easy to see that the collection of all relative averaging conformal algebras and homomorphisms between them is a category. We denote this category by ${\mathbf{rAvgC}}$.

\begin{exam}
Let $\mathcal{A}$ be an associative conformal algebra and $P: \mathcal{A} \to \mathcal{A}$ be an averaging operator on it. Then $P$ can be regarded as a relative averaging operator on $\mathcal{A}$ with respect to the adjoint $\mathcal{A}$-bimodule. Hence, a relative averaging operator generalizes an averaging operator.
\end{exam}

\begin{exam}
Let $\mathcal{A}$ be an associative conformal algebra and $\mathcal{M}$ be a conformal $\mathcal{A}$-bimodule. Suppose $P:\mathcal{M} \rightarrow \mathcal{A}$ is a conformal $\mathcal{A}$-bimodule map. Then $P$ is a relative averaging operator on $\mathcal{A}$ with respect to $\M$, equivalently, $(\mathcal{A},\mathcal{M}, P)$ is a relative averaging conformal algebra.
\end{exam}

 \begin{exam}
     Let $\mathcal{A}$ be an associative conformal algebra. Then the $\C[\p]$-module $\underbrace{\mathcal{A}\oplus\cdots\oplus\mathcal{A}}_{n\text{ copies}}$ can be given a conformal $\mathcal{A}$-bimodule structure with the left and right $\la$-actions given by
     \begin{align*}
         a_{\la}(a_1, \ldots, a_n ) = (a_{\la}a_1, \ldots , a_{\la}a_{n}) ~~~  \text{ and } ~~~ (a_1, \ldots , a_{n})_{\la}a = ({a_1}_{\la}a, \ldots , {a_{n}}_{\la}a ), \text{ for } a, a_1, \ldots , a_n \in \mathcal{A}.
     \end{align*}
     Then it is easy to see that the map
     \begin{align*}
         P : \mathcal{A} \oplus \cdots \oplus \mathcal{A} \longrightarrow \mathcal{A}, \quad P\big( (a_1, \ldots, a_n) \big) = a_1 + \cdots + a_n, \text{ for }  a_1, \ldots , a_n \in \mathcal{A}
     \end{align*}
     is a relative averaging operator. In other words, $(\mathcal{A},\, \mathcal{A} \oplus \cdots \oplus \mathcal{A},\,P)$ is a relative averaging conformal algebra.
 \end{exam}

\begin{exam}
Let $\mathcal{A}$ be an associative conformal algebra and let $G$ be a finite group acting on $\mathcal{A}$ by conformal algebra automorphisms. 
Define a $\mathbb{C}[\partial]$-linear map $P:\mathcal{A}\rightarrow\mathcal{A}$ by $P(a) :=\sum \limits_{g\in G}g\cdot a$, for $a \in \mathcal{A}$. Then, for any $a,b \in \mathcal{A}$, we observe that
\begin{align*}
&P ( P(a) _\la b ) = \sum_{h \in G} h \cdot \big( \big( \sum_{g \in G} g \cdot a \big) _{\la} b \big) = \sum_{h ,g\in G} h \cdot \big ( (g \cdot a)_{\la}b \big) = \big(  \sum_{g \in G} g \cdot a \big)  _\la \big( \sum_{h \in G} h \cdot b  \big) = P(a) _\la P(b),\\
&P ( a _{\la} P(b) ) = \sum_{h \in G} h \cdot \big( a _{\la} \big(\sum_{g \in G} g \cdot b \big)  \big) = \sum_{h ,g \in G} h \cdot \big ( a _{\la} ( g \cdot b )\big) = \big( \sum_{h \in G} h \cdot a \big) _{\la} \big( \sum_{g \in G} g \cdot b \big) = P(a) _{\la} P(b).
\end{align*}
This shows that $P: \mathcal{A} \to \mathcal{A}$ is an averaging operator on $\mathcal{A}$. 
\end{exam}

We will now provide a characterization of (relative) averaging operators.
 First, we start with the following useful result.
 
\begin{proposition}\label{new-diass}
    Let $\mathcal{A}$ be an associative conformal algebra and $\M$ be a conformal $\mathcal{A}$-bimodule. Then the $\C[\p]$-module $\mathcal{A} \oplus \M$ inherits a diassociative conformal algebra structure with the $\la$-products:
     \begin{align*}
         (a,u)\dashv_{\la}(b,v) = (a_{\la}b, u_{\la}b) ~~~ \text{ and } ~~~ (a,u)\vdash_{\la}(b,v) = (a_{\la}b, a_{\la}v), \text{ for } a,b\in \mathcal{A}, u,v \in \M.
     \end{align*}
\end{proposition}
\begin{proof}
    It is easy to see that both the $\la$-products defined above are conformal sesquilinear. Next, for any $(a, u),(b, v),(c, w) \in \mathcal{A} \oplus \M$, we have
    \begin{align*}     
 & (a, u) \dashv_{\la} \big ( (b, v)  \dashv_{\mu} (c, w) \big)= (a,u) \dashv_{\la} (b_{\mu}c, v_{\mu}c)=\big( a_{\la}(b_{\mu}c), u_{\la}(b_{\mu}c) \big),\\&
\big( (a, u) \dashv_{\la}  (b, v)  \big ) \dashv_{\la + \mu} (c, w) = (a_{\la}b,u_{\la}b) \dashv_{\la + \mu} (c, w)=\big( (a_{\la}b)_{\la+\mu}c, (u_{\la}b)_{\la+\mu}c \big),\\&
 (a, u) \dashv_{\la} \big ( (b, v)  \vdash_{\mu} (c, w) \big)= (a,u) \dashv_{\la} (b_{\mu}c, b_{\mu}w)=\big( a_{\la}(b_{\mu}c), u_{\la}(b_{\mu}c) \big).
    \end{align*}
    Thus, by the conformal associativity of $\mathcal{A}$ and the conformal $\mathcal{A}$-bimodule property of $\M$, we obtain
    \begin{align*}
        (a, u) \dashv_{\la} \big ( (b, v)  \dashv_{\mu} (c, w) \big)=\big( (a, u) \dashv_{\la}  (b, v)  \big ) \dashv_{\la + \mu} (c, w) =(a, u) \dashv_{\la} \big ( (b, v)  \vdash_{\mu} (c, w) \big).
    \end{align*}
    Similarly, one can show that
    \begin{align*}
    & \qquad \qquad  \qquad \quad (a, u) \vdash_{\la} \big ( (b, v)  \dashv_{\mu} (c, w) \big)=\big( (a, u) \vdash_{\la}  (b, v)  \big ) \dashv_{\la + \mu} (c, w),\\&
         (a, u) \vdash_{\la} \big ( (b, v)  \vdash_{\mu} (c, w) \big)=\big( (a, u) \vdash_{\la}  (b, v)  \big ) \vdash_{\la + \mu} (c, w) =\big( (a, u) \dashv_{\la}  (b, v)  \big ) \vdash_{\la + \mu} (c, w).
    \end{align*}
    This completes the proof.
\end{proof}

\begin{proposition}\label{prop-graph}
    Let $\mathcal{A}$ be an associative conformal algebra and $\M$ be a conformal $\mathcal{A}$-bimodule. Then a $\C[\p]$-linear map $P: \M \to \mathcal{A}$ is a relative averaging operator on $\mathcal{A}$ with respect to $\M$ if and only if the graph $Gr(P) = \{(P(u),u) :  u\in \M\}$ is a subalgebra of the diassociative conformal algebra $\mathcal{A} \oplus \M$ considered in Proposition \ref{new-diass}.
\end{proposition}

\begin{proof}
  For any $u,v \in \M$, we observe that
  \begin{align*}
  (P(u),u) \dashv_{\la} (P(v),v)=(P(u)_{\la}P(v), u_{\la}P(v)).
  \end{align*}
  This is in ${Gr}(P)[\la]$ if and only if we have $P(u)_{\la}P(v)= P(u_{\la}P(v))$. Similarly, we have that the $\la$-product $(P(u),u) \vdash_{\la} (P(v),v)=(P(u)_{\la}P(v), P(u)_{\la}v)$ is in ${Gr}(P)[\la]$ if and only if $P(u)_{\la}P(v)= P(P(u)_{\la}v)$. In other words, ${Gr}(P)$ is a subalgebra of the diassociative conformal algebra $\mathcal{A} \oplus \M$ if and only if $P$ is a relative averaging operator.
\end{proof}

\subsection*{Functorial relations with diassociative conformal algebras.}
\begin{proposition} \label{prop*}
    Let $(\mathcal{A}, \M, P)$ be a relative averaging conformal algebra. Then the $\C[\p]$-module $\M$ carries a diassociative conformal algebra structure with the $\la$-products
    \begin{align}\label{induced-di}
        u \dashv_{\la} v := u_{\la} P(v) ~~~ \text{ and } ~~~ u \vdash_{\la} v := P(u)_{\la} v,  \text{ for } u, v \in \M.
    \end{align}
We denote this diassociative conformal algebra by $\M_{P}$.
\end{proposition}

\begin{proof}
Since $\M$ is isomorphic to ${Gr}(P) = \{(P(u),u) :  u\in \M\}$ by a $\C [\partial]$-linear isomorphism, it follows from Proposition \ref{prop-graph} that $\M$ carries a diassociative conformal algebra structure. This structure is precisely given by (\ref{induced-di}).
\end{proof}

The above construction is functorial. Indeed, if
$(f,g):(\mathcal{A},\mathcal{M},P)\rightarrow(\mathcal{A}',\mathcal{M}',P')$
is a homomorphism of relative averaging conformal algebras, then the map $g: \mathcal{M}_{P}\rightarrow\mathcal{M}'_{P'}$ is a homomorphism of the induced diassociative conformal algebras. Consequently, there is a functor $\mathcal{F}:\mathbf{rAvgC} \rightarrow\mathbf{DiassC}$, which assigns to each relative averaging conformal algebra
$(\mathcal{A},\mathcal{M},P)$ the induced diassociative conformal algebra
$\mathcal{M}_{P}$, and to each relative averaging conformal algebra homomorphism $(f,g)$ the homomorphism $g$. We next construct a functor in the opposite direction.

\begin{thm}\label{thm-diass-to-avg}
For any diassociative conformal algebra $(\mathcal{D},\dashv_{\la},\vdash_{\la})$, there exists a relative averaging conformal algebra $(\mathcal{A}, \mathcal{D}, P)$ so that the induced diassociative conformal algebra $ \mathcal{D}_P$, as in Proposition~\ref{prop*}, coincides with the prescribed diassociative conformal algebra structure.
\end{thm}

\begin{proof}
Let $\mathcal{I}$ be the ideal of the diassociative conformal algebra $\mathcal{D}$ generated by the set
\begin{align*}
\{\,a\dashv_{(j)}b-a\vdash_{(j)}b
\mid a,b\in \mathcal{D}, ~ \! j\ge0\,\}.
\end{align*}
Then the quotient $\mathcal{A}:=\mathcal{D}/\mathcal{I}$ is an associative conformal algebra with the $\la$-product
\begin{align*}
[a]_{\la}[b]
:=
\sum_{j\ge0}[a\dashv_{(j)}b] ~ \! \frac{\la^j}{j!}
=
\sum_{j\ge0}[a\vdash_{(j)}b] ~ \!\frac{\la^j}{j!}, \text{ for } a, b \in \mathcal{D}.
\end{align*}
Furthermore, $\mathcal{D}$ becomes a conformal $\mathcal{A}$-bimodule via the left and right $\la$-actions
\begin{align*}
[a]_{\la}b=a\vdash_{\la}b
~~~ \text{ and } ~~~
b_{\la}[a]=b\dashv_{\la}a,  \text{ for } a, b \in \mathcal{D}.
\end{align*}
Let $P:\mathcal{D} \rightarrow \mathcal{A}$, $P(a)=[a]$
be the canonical quotient map. Then, for all $a,b\in\mathcal{D}$, we have
\begin{align*}
P(a)_{\la}P(b)
=[a]_{\la}[b] 
=\sum_{j\ge0}[a\dashv_{(j)}b] ~ \! \frac{\la^j}{j!}
=P(a_{\la}P(b)),\\
P(a)_{\la}P(b)
=[a]_{\la}[b] 
=\sum_{j\ge0}[a\vdash_{(j)}b] ~ \! \frac{\la^j}{j!}
=P(P(a)_{\la}b).
\end{align*}
Hence, $P$ is a relative averaging operator on $\mathcal{A}$ with respect to the conformal $\mathcal{A}$-bimodule $\D$.

Finally, by Proposition~\ref{prop*}, if $\D_P = (\D, \dashv'_\la, \vdash'_\la)$ is the induced diassociative conformal algebra then
\begin{align*}
a\dashv'_{\la}b :=a_{\la}P(b) = a_\la [b] = a \dashv_\la b ~~~ \text{ and }
~~~ a\vdash'_{\la}b := P(a)_{\la}b = [a]_\la b = a \vdash_\la b, \text{ for } a, b \in \D.
\end{align*}
Hence the result follows.
\end{proof}

Let $\mathcal{D}$ and
$\mathcal{D}'$
be two diassociative conformal algebras and $ \varphi:\mathcal{D} \rightarrow \mathcal{D}'$
be a homomorphism of diassociative conformal algebras. Then it is easy to verify that the pair $(f,g = \varphi)$ is a homomorphism of relative averaging conformal algebras from
$(\mathcal{A},\mathcal{D},P)$ to $(\mathcal{A}',\mathcal{D}',P')$, where $f:\mathcal{A}=\mathcal{D}/\mathcal{I} \rightarrow \mathcal{A}'=\mathcal{D}'/\mathcal{I}'$ is given by $f([a])=[ \varphi(a)]$, for $a \in \mathcal{D}$. This construction yields a functor $\mathcal{G}:\mathbf{DiassC} \rightarrow\mathbf{rAvgC}$.

\begin{proposition}\label{prop-left}
    The functor $\GG$ constructed above is left adjoint to the functor $\mathcal{F}$.
Equivalently, for every diassociative conformal algebra $(\D,\dashv_{\la},\vdash_{\la})$ and every relative averaging conformal algebra $(\mathcal{A'},\mathcal{M'}, P')$, there is a natural bijection
\[
\operatorname{Hom}_{\mathbf{DiassC}}\! (\D,\, \mathcal{M'}_{P'} )
\;\cong\;
\operatorname{Hom}_{\mathbf{rAvgC}}\! ((\mathcal{D}/\mathcal{I},\D, P),\,
(\mathcal{A'},\mathcal{M'}, P') ).
\]
\end{proposition}

\begin{proof}
Let $\varphi \in \operatorname{Hom}_{\mathbf{DiassC}}\! (\D,\, \mathcal{M'}_{P'} )$. Then we define a $\C[\p]$-linear map $f : \mathcal{D}/\mathcal{I} \to \mathcal{A}'$ by $ f ([a]) := P'( \varphi (a)),\text{ for } [a]\in  \mathcal{D}/\mathcal{I}$. For any $[a],[b]\in \mathcal{D}/\mathcal{I}$, we have
\begin{align*}
f ([a]_{\la}[b])
=f \big( \sum_{j\geq 0}[a \dashv_{(j)}b] ~ \! \frac{\la^j}{j!} \big)
=P' (\varphi(a) \dashv_{\la} \varphi(b) )
=P'(\varphi(a))_{\la} P'( \varphi(b))
=f([a])_{\la} f([b]).
\end{align*}
This shows that $f$ is a homomorphism of associative conformal algebras. Moreover, for any $[a] \in \mathcal{D}/\mathcal{I}$ and $b \in \D$, we have
\begin{align*}
&\varphi ([a]_{\la}b)
=\varphi (a \vdash_{\la} b)
=\varphi(a) \vdash_{\la} \varphi(b)
=P' (\varphi(a))_{\la} \varphi(b)
=f([a])_{\la} \varphi(b),\\
&\varphi (b_{\la}[a])
=\varphi (b \dashv_{\la} a)
=\varphi(b) \dashv_{\la} \varphi(a)
=\varphi(b)_{\la} P' (\varphi(a))
= \varphi(b) _{\la} f([a]).
\end{align*}
Furthermore, we have $f \circ P = P' \circ \varphi$. Hence, $(f,\varphi) \in \operatorname{Hom}_{\mathbf{rAvgC}}\! ((\mathcal{D}/\mathcal{I},\D, P),\,
(\mathcal{A'},\mathcal{M'}, P') )$.

\medskip

On the other hand, if $(f,\varphi) \in \operatorname{Hom}_{\mathbf{rAvgC}}\! ((\mathcal{D}/\mathcal{I},\D, P),\,
(\mathcal{A'},\mathcal{M'}, P') )$, then for any $a,b \in \D$, we have
\begin{align*}
    &\varphi( a \dashv_\la b) = \varphi(a _{\la} [b]) = \varphi(a)_{\la} f([b])= \varphi(a)_{\la} f(P(b)) = \varphi(a)_{\la}P'(\varphi(b))= \varphi(a) \dashv_\la \varphi(b),\\
    & \varphi( a \vdash_\la b) = \varphi([a] _{\la} b) = f([a])_{\la} \varphi(b)= f(P(a))_{\la}\varphi(b) = P'(\varphi(a))_{\la}\varphi(b)= \varphi(a) \vdash_\la \varphi(b).
\end{align*}
This shows that $\varphi \in \operatorname{Hom}_{\mathbf{DiassC}}\! (\D,\, \mathcal{M'}_{P'} )$. Also, one can easily verify that these two correspondences are inverses of each other.
\end{proof}

\subsection*{Maurer-Cartan characterization and cohomology of relative averaging operators} Here we introduce the cohomology of a relative averaging operator defined on an associative conformal algebra with respect to a conformal bimodule. 

\medskip

Let $\mathcal{A}$ be an associative conformal algebra and $\M$ be a conformal $\mathcal{A}$-bimodule. We endow the direct sum $\D= \mathcal{A} \oplus \M$ with the diassociative conformal algebra structure given in Proposition~\ref{new-diass}. Next, we consider the graded Lie algebra $ \mathfrak{g} =\big( \bigoplus \limits_{n=0}^{\infty}C_{cY}^{n+1}(\D,\D),[\,,\,] \big)$ associated to the $\C [\partial]$-module $\D$.
Note that the above diassociative conformal algebra structure on $\D$ determines a Maurer-Cartan element $\pi \in C_{cY}^{2}(\D, \D)$, that is, $[\pi, \pi]=0$. Explicitly, $\pi$ is given by 
\begin{align*}
    &\pi_ \la \big (\tikz[baseline=-0.5ex,scale=0.20]{
  \draw (-1,1)--(0,0);
  \draw (1,1)--(0,0);
  \draw (0,0)--(0,-1);
  \draw (0.5,0.5)--(0,1);
}; (a,u), (b,v)\big)  := (a,u) \dashv_{\la } (b,v) = ( a _ \la b, u _\la b),\\
& \pi_ \la \big( \tikz[baseline=-0.5ex,scale=0.20]{
  \draw (-1,1)--(0,0);
  \draw (1,1)--(0,0);
  \draw (0,0)--(0,-1);
  \draw (-0.5,0.5)--(0,1);
}; (a,u), (b,v) \big):= (a,u)\vdash_{\la} (b,v) = (a _ \la b, a_\la v),
\end{align*}
for $(a,u), (b,v) \in \D$.  Moreover, it is easy to see that the graded subspace $\mathfrak{h}=\bigoplus \limits_{n = 0} ^{\infty}C_{cY}^{n+1}(\M,\mathcal{A})$ is an abelian graded Lie subalgebra of $\mathfrak{g}$. Let $p: \mathfrak{g} \rightarrow \mathfrak{g}$ be the projection onto the subspace $\mathfrak{h}$. Then $\rm{ker}(p) \subset \mathfrak{g}$ is a graded Lie subalgebra and $\rm{im}(p)=\mathfrak{h}$. Moreover, $\pi \in \rm{ker}(p)$ and hence, by Kosmann-Schwarzbach’s derived bracket construction \cite{kosmann}, the shifted graded space $s\mathfrak{h}=\bigoplus \limits_{n = 1} ^{\infty}C_{cY}^{n}(\M,\mathcal{A})$ carries a graded Lie algebra structure with the bracket 
 \begin{align*}
\{\! \! \{ f, g \}\! \! \} := (-1)^{m} ~ \! [[\pi, f],g],
\end{align*}
for $f \in  C_{cY} ^{m }(\M,\mathcal{A})$ and $g \in  C_{cY} ^{n }(\M,\mathcal{A})$. 
Explicitly, this bracket is given by
\begin{align} \label{equation x}
 &{\{\! \! \{ f, g \}\! \! \}}_{\la_{1}, \ldots ,\la_{m+n-1}} \left(y ; u_{1}, \ldots, u_{m+n} \right) \nonumber \\
& = \sum_{i=1}^{m}(-1)^{(i-1) n} f_{\la_{1}, \ldots , \la_{i-1}, \la_{i} +\cdots +\la_{i+n}, \la_{i+n+1}, \ldots, \la_{m+n-1}} \big( R_{0}^{m ; i, n+1}(y) ; u_{1}, \ldots, u_{i-1}, \nonumber \\
& \quad \pi_{\la_{i}+\cdots +\la_{i+n-1}}\big(R_{0}^{2 ; 1, n} R_{i}^{m; i, n+1}(y) ; g_{\la_{i}, \ldots, \la_{i+n-2}}\big(R_{1}^{2 ; 1, n} R_{i}^{m ; i, n+1}(y) ; u_{i}, \ldots, u_{i+n-1}\big), u_{i+n}\big), \ldots, u_{m+n} \big)  \nonumber \\
& \quad -\sum_{i=1}^{m}(-1)^{i n} ~ \! f_{\la_{1}, \ldots , \la_{i-1}, \la_{i} +\cdots +\la_{i+n}, \la_{n+i+1}, \ldots, \la_{m+n-1}} \big(R_{0}^{m ; i, n+1}(y) ; u_{1}, \ldots, u_{i-1},  \nonumber\\
& \quad \quad \pi_{\la_{i}+ \cdots +\la_{i+n-1}} \big(R_{0}^{2 ; 2, n} R_{i}^{m ; i, n+1}(y) ; u_{i}, g_{\la_{i}, \ldots ,\la_{i+n-2}} \big(R_{2}^{2 ; 2, n} R_{i}^{m ; i, n+1}(y) ; u_{i+1}, \ldots, u_{i+n}\big)\big),\ldots, u_{m+n}\big) \nonumber \\
&\quad \quad -(-1)^{m n} \sum_{i=1}^{n}(-1)^{(i-1) m} ~ \! g_{\la_{1}, \ldots , \la_{i-1},\la_{i}+\cdots+\la_{i+m},\la_{i+m+1}, \ldots ,\la_{m+n-1}} \big(R_{0}^{n ; i, m+1}(y) ; u_{1}, \ldots, u_{i-1},\nonumber \\
& \quad \pi_{\la_{i}+\cdots+\la_{i+m-1}} \big(R_{0}^{2 ; 1, m} R_{i}^{n ; i, m+1}(y) ; f_{\la_{i}, \ldots ,\la_{i+m-2}} \big(R_{1}^{2 ; 1, m} R_{i}^{n ; i, m+1}(y) ; u_{i}, \ldots, u_{i+m-1}\big), u_{i+m} \big), \ldots, u_{m+n} \big) \nonumber \\
& \quad \quad +(-1)^{m n} \sum_{i=1}^{n}(-1)^{i m} g_{\la_{1}, \ldots , \la_{i-1},\la_{i}+\cdots+\la_{i+m},\la_{i+m+1}, \ldots ,\la_{m+n-1}} \big(R_{0}^{n ; i, m+1}(y) ; u_{1}, \ldots, u_{i-1}, \nonumber \\
& \quad \quad  \pi_{\la_{i}+\cdots+\la_{i+m-1}} \big(R_{0}^{2 ; 2, m} R_{i}^{n ; i, m+1}(y) ; u_{i}, f_{\la_{i},\ldots,\la_{i+m-2}}\big(R_{2}^{2 ; 2, m} R_{i}^{n ; i, m+1}(y) ; u_{i+1}, \ldots, u_{i+m}\big)\big), \ldots, u_{m+n} \big)  \nonumber \\
&\quad \quad +(-1)^{m n} \pi_{\la_{1}+\cdots+\la_{m}} \big(R_{0}^{2 ; 1, m} R_{0}^{m+1 ; m+1, n}(y) ; \nonumber \\
& \quad \quad  f_{\la_{1},\ldots,\la_{m-1}} \big(R_{1}^{2 ; 1, m} R_{0}^{m+1 ; m+1, n}(y) ;u_{1}, \ldots, u_{m}\big), g_{\la_{m+1},\ldots,\la_{m+n-1}}\big(R_{m+1}^{m+1 ; m+1, n}(y) ; u_{m+1}, \ldots, u_{m+n}\big) \big) \nonumber \\
& \quad -\pi_{\la_{1}+\cdots+\la_{n}}\big(R_{0}^{2 ; 1, n} R_{0}^{n+1 ; n+1, m}(y) ; \nonumber \\
& \quad \quad  g_{\la_{1},\dots,\la_{n-1}} \big(R_{1}^{2 ; 1, n} R_{0}^{n+1 ; n+1, m}(y) ; u_{1}, \ldots, u_{n}\big), f_{\la_{n+1}, \ldots, \la_{m+n-1}} \big(R_{n+1}^{n+1 ; n+1, m}(y) ; u_{n+1}, \ldots, u_{m+n}\big) \big), 
\end{align}
for $y \in Y_{m+n}$ and $u_{1}, \ldots, u_{m+n} \in \M$.
Then we have the following result.

\begin{thm}\label{thm-mc-avg}
    Let $\mathcal{A}$ be an associative conformal algebra and $\M$ be a conformal $\mathcal{A}$-bimodule. Then the pair $\big( \bigoplus \limits_{n = 1} ^{\infty} C_{cY}^{n}(\M,\mathcal{A}),\{\!\! \{ \, , \,  \}\! \! \}\big)$ is a graded Lie algebra. Moreover, a $\C[\p]$-linear map $P : \M \to \mathcal{A}$ is a relative averaging operator if and only if $P \in  C_{cY} ^{1}(\M,\mathcal{A})$ is a Maurer-Cartan element of the graded Lie algebra $\big( \bigoplus \limits_{n = 1} ^{\infty} C_{cY}^{n}(\M,\mathcal{A}), \{\!\! \{ \, , \,  \}\! \! \} \big)$.
\end{thm}

\begin{proof}
    The first part follows from the previous discussion. To prove the second part, we observe that $C^{1}_{cY}(\M,\mathcal{A}) \cong \mathrm{Hom}_{\C[\p]}(\M,\mathcal{A})$. Therefore, any $\C[\p]$-linear map $P: \M \to \mathcal{A}$ can be regarded as an element of $ C^1(\M,\mathcal{A})$ via the identification $P ( \tikz[baseline=-0.5ex,scale=0.20]{
  \draw (-1,1)--(0,0);
  \draw (1,1)--(0,0);
  \draw (0,0)--(0,-1);
}; u )= P(u)$, for $u \in \M$.
Then, it follows from (\ref{equation x}) that
\begin{align*}
\frac{1}{2} {\{\! \! \{ P, P \}\! \! \}}_{\la} (  \tikz[baseline=-0.5ex,scale=0.20]{
  \draw (-1,1)--(0,0);
  \draw (1,1)--(0,0);
  \draw (0,0)--(0,-1);
  \draw (-0.5,0.5)--(0,1);
}; u, v ) & =  P \big( \underbrace{\pi_{\la} ( \tikz[baseline=-0.5ex,scale=0.20]{
  \draw (-1,1)--(0,0);
  \draw (1,1)--(0,0);
  \draw (0,0)--(0,-1);
  \draw (-0.5,0.5)--(0,1);
} ; P(u), v )}_{=P(u)_{\la}v}\big)+P \big( \underbrace{\pi_{\la}(\tikz[baseline=-0.5ex,scale=0.20]{
  \draw (-1,1)--(0,0);
  \draw (1,1)--(0,0);
  \draw (0,0)--(0,-1);
  \draw (-0.5,0.5)--(0,1);
} ; u, P(v) )}_{=0} \big)-\underbrace{\pi_{\la} (\tikz[baseline=-0.5ex,scale=0.20]{
  \draw (-1,1)--(0,0);
  \draw (1,1)--(0,0);
  \draw (0,0)--(0,-1);
  \draw (-0.5,0.5)--(0,1);
} ; P(u), P(v) )}_{P(u) _{\la} P(v)}\\&
= P (P(u)_{\la}  v) -P(u)_{\la} P(v) ,
\end{align*}
\begin{align*}
\frac{1}{2} {\{\! \! \{ P, P \}\! \! \}}_{\la} ( \tikz[baseline=-0.5ex,scale=0.20]{
  \draw (-1,1)--(0,0);
  \draw (1,1)--(0,0);
  \draw (0,0)--(0,-1);
  \draw (0.5,0.5)--(0,1);
} ; u, v ) &=  P \big(\underbrace{\pi_{\la} ( \tikz[baseline=-0.5ex,scale=0.20]{
  \draw (-1,1)--(0,0);
  \draw (1,1)--(0,0);
  \draw (0,0)--(0,-1);
  \draw (0.5,0.5)--(0,1);
} ; P(u), v )}_{=0}\big) + P \big(\underbrace{\pi_{\la} ( \tikz[baseline=-0.5ex,scale=0.20]{
  \draw (-1,1)--(0,0);
  \draw (1,1)--(0,0);
  \draw (0,0)--(0,-1);
  \draw (0.5,0.5)--(0,1);
} ; u, P(v) )}_{=u_{\la} P(v)} \big) - \underbrace{\pi_{\la} (\tikz[baseline=-0.5ex,scale=0.20]{
  \draw (-1,1)--(0,0);
  \draw (1,1)--(0,0);
  \draw (0,0)--(0,-1);
  \draw (0.5,0.5)--(0,1);
} ;P(u), P(v) )}_{P(u) _{\la} P(v)} \\&
= P  (u _{\la}  P(v) ) -P(u) _{\la} P(v),
\end{align*}
for $u, v \in \M$. Hence, $P$ is a Maurer-Cartan element (i.e. $\{\! \! \{ P, P \} \! \!\} = 0$) if and only if $P$ is a relative averaging operator on $\mathcal{A}$ with respect to $\M$.
\end{proof}

\medskip
A relative averaging operator $P$ induces a map $\delta_P : C_{cY}^{n}(\M, \mathcal{A}) \rightarrow C_{cY}^{n+1}(\M, \mathcal{A})$, for $n \geq 1$, given by
\begin{align*}
    \delta_P (f) = \{\! \! \{ P, f \} \! \! \}, \text{ for } f \in  C_{cY} ^{n \geq 1} (\M,\mathcal{A}).
\end{align*}
Then it follows from the above theorem that $\delta_P^2 = 0$. Hence, $\{ C^\bullet_{cY} (\M, \mathcal{A}), \delta_P \}$ is a cochain complex. Next, we show that the map $\delta_P$ can be seen as the coboundary operator of the induced diassociative conformal algebra $\M_P$ with coefficients in a suitable representation.
We begin with the following straightforward result.

\begin{proposition}\label{ind-rep}
Let $\mathcal{A}$ be an associative conformal algebra, $\M$ be a conformal $\mathcal{A}$-bimodule, and $P: \M \to \mathcal{A}$ be a relative averaging operator. Then the induced diassociative conformal algebra $\M_P$ has a representation on $\mathcal{A}$ with the left and right $\la$-actions given by
    \begin{align*}
        &\dashv_{\la} ~ \! : \M \times \mathcal{A} \to \mathcal{A}[\la], \quad u \dashv_{\la} a := P(u)_{\la}a - P(u_{\la} a),\\
       & \vdash_{\la} ~ \!:\ M \times \mathcal{A} \to \mathcal{A}[\la], \quad u \vdash_{\la} a := P(u)_{\la}a, \\
        &\dashv_{\la} ~ \! : \mathcal{A} \times  \M \to \mathcal{A}[\la], \quad a \dashv_{\la} u := a_{\la} P(u),\\ 
        &\vdash_{\la} ~ \! :\mathcal{A} \times  \M \to \mathcal{A} [\la], \quad a \vdash_{\la} u := a_{\la} P(u) - P(a_{\la} u).
    \end{align*}
\end{proposition}

\begin{proposition}\label{two-coboundary}
    Let $\mathcal{A}$ be an associative conformal algebra, $\M$ be a conformal $\mathcal{A}$-bimodule, and $P: \M \to \mathcal{A}$ be a relative averaging operator. Then for any $f \in C_{cY}^{n \geq 1} (\M, \mathcal{A})$, we have
    \begin{align*}
        \delta_P (f) = (-1)^n ~ \! \delta_{cDiass} (f),
    \end{align*}
    where $\delta_{cDiass}$ is the coboundary map of the induced diassociative conformal algebra $\M_P$ with coefficients in the representation given in Proposition \ref{ind-rep}.
\end{proposition}

\begin{proof}
For any $y\in Y_{n+1}$ and $u_1,\ldots,u_{n+1}\in \mathcal{M}$, we have
\begin{align*}
&\delta_P(f)_{\la_1, \ldots, \la_n}(y;u_1,\ldots,u_{n+1})\\
=~& \{\! \! \{ P, f \}\! \! \}_{\la_1, \ldots, \la_n}(y;u_1,\ldots,u_{n+1})\\
=~& P\big(\pi_{\la_1 + \cdots + \la_n} \big(R_0^{2;1,n}(y);f_{\la_1, \ldots, \la_{n-1}} (R_1^{2;1,n}(y);u_1,\ldots,u_n ),u_{n+1} \big)\big)\\
& -(-1)^n P \big(\pi_{\la_1 + \cdots + \la_n} \big( R_0^{2;2,n}(y);u_1,
f_{\la_1, \ldots, \la_{n-1}} (R_2^{2;2,n}(y);u_2,\ldots,u_{n+1} ) \big) \big)\\
& -(-1)^n\sum_{i=1}^{n}(-1)^{i-1}
f_{\la_1, \ldots, \la_{i-1}, \la_i + \la_{i+1}, \la_{i+2}, \ldots, \la_{n} }(R_0^{n;i,2}(y);u_1,\ldots,
\pi_{\la_i} (R_i^{n;i,2}(y);P(u_i),u_{i+1} ),\ldots,u_{n+1} )\\
& +(-1)^n\sum_{i=1}^{n}(-1)^i
f_{\la_1, \ldots, \la_{i-1}, \la_i + \la_{i+1}, \la_{i+2}, \ldots, \la_{n} }(R_0^{n;i,2}(y);u_1,\ldots,
\pi_{\la_i} (R_i^{n;i,2}(y);u_i,P(u_{i+1}) ),\ldots,u_{n+1} )\\
& +(-1)^n ~ \! \pi_{\la_1} (R_0^{2;2,n}(y);
P(u_1),f_{\la_2, \ldots, \la_{n}} (R_2^{2;2,n}(y);u_2,\ldots,u_{n+1} ) )\\
& -\pi_{\la_1+ \cdots+\la_n} (R_0^{2;1,n}(y);
f_{\la_1,\ldots,\la_{n-1}} (R_1^{2;1,n}(y);u_1,\ldots,u_n ),P(u_{n+1}) )\\
=~&(-1)^n \bigg\{
\pi_{\la_1}\big(R_0^{2;2,n}(y);
P(u_1),f_{\la_2, \ldots, \la_n}(R_2^{2;2,n}(y);u_2,\ldots,u_{n+1})\big)\\
&
-P \big( \pi_{\la_1 + \cdots +\la_n}(R_0^{2;2,n}(y);
u_1,f_{\la_1, \ldots, \la_{n-1}}(R_2^{2;2,n}(y);u_2,\ldots,u_{n+1}) ) \big)\\
& +\sum_{i=1}^{n}(-1)^i ~ \! f_{\la_1, \ldots, \la_{i-1}, \la_i + \la_{i+1}, \la_{i+2}, \ldots, \la_{n} } \big( R_0^{n;i,2}(y);u_1,\ldots,u_{i-1}, \pi_{\la_i} (R_i^{n;i,2}(y);P(u_i),u_{i+1} )\\
&\qquad \qquad \qquad \qquad \qquad \qquad \qquad \qquad \qquad \qquad \qquad +\pi_{\la_i} (R_i^{n;i,2}(y);u_i,P(u_{i+1}) ),
u_{i+2},\ldots,u_{n+1}\big)\\
& +(-1)^{n+1} ~  \! \pi_{\la_1+ \cdots+\la_n}\left(R_0^{2;1,n}(y);
f_{\la_1,\ldots,\la_{n-1}}\left(R_1^{2;1,n}(y);u_1,\ldots,u_n\right),P(u_{n+1})\right)\\
& -(-1)^{n+1} P \big( \pi_{\la_1 + \cdots + \la_n} (R_0^{2;1,n}(y);f_{\la_1, \ldots, \la_{n-1}} (R_1^{2;1,n}(y);u_1,\ldots,u_n) ,u_{n+1} ) \big)  \bigg\} 
\end{align*}
\begin{align*}
=~&(-1)^n \bigg\{
u_1(\bullet_0^y)_{\la_1} f_{\la_2, \ldots , \la_n}(d_{0}y;u_2,\ldots,u_{n+1})\\
& \qquad \quad +\sum_{i=1}^{n}(-1)^i
f_{ \la_{1}, \ldots , \la_{i-1}, \la_{i}+ \la_{i+1}, \la_{i+2}, \ldots , \la_{n}}\left(d_{i}y;u_1,\ldots,u_{i-1},
u_i(\bullet_i^y)_{\la_i} u_{i+1},\ldots,u_{n+1}\right)\\
&  \qquad \quad +(-1)^{n+1}f_{\la_1, \ldots, \la_{n-1}}(d_{n+1}y;u_1,\ldots,u_n)
(\bullet_{n+1}^y)_{\la_1 + \cdots + \la_n} u_{n+1}
\bigg\} \\
& \quad \qquad \quad \qquad \quad \qquad \quad  \text{(as $R_2^{2;2,n}(y)=d_0(y)$,~$R_0^{n;i,2}(y)=d_i(y)$ and $R_1^{2;1,n}(y)=d_{n+1}(y)$)}\\
=~&(-1)^n ~ \! \delta_{cDiass}(f)_{\la_1, \ldots, \la_n}(y; u_1,\ldots,u_{n+1}).
\end{align*}
This completes the proof.
\end{proof}

Let $P: \M \rightarrow \mathcal{A}$ be a relative averaging operator. Then the cohomology of the induced diassociative conformal algebra $\M_P$ with coefficients in the representation $\mathcal{A}$ (given in Proposition \ref{ind-rep}) is defined to be the {\bf cohomology} of the relative averaging operator $P$. Then it follows from Proposition \ref{two-coboundary} that, for any $n \geq 2$, the $n$-th cohomology group of the relative averaging operator $P$ is isomorphic to the $n$-th cohomology group of the cochain complex $\{ C^\bullet_{cY} (\M, \mathcal{A}), \delta_P \}$. We note that the cohomology of a relative averaging operator can be used to study the deformations of the operator. In a subsequent paper, we shall study (relative) averaging operators on free conformal algebras \cite{roitman-1,free-roit} and compute the cohomology groups.

\medskip

\noindent {\bf Acknowledgements.} 
Anupam Sahoo would like to thank PMRF, Government of India, for funding the PhD fellowship. Both authors thank the Department of Mathematics, IIT Kharagpur, for providing a beautiful academic atmosphere in which the research has been carried out.

\medskip

\noindent {\bf Data Availability Statement.} Data sharing does not apply to this article as no new data were created or analyzed in this study.

\end{document}